\relax
\documentclass[letterpaper]{article} %
\usepackage[margin=1.2in]{geometry}
\usepackage{times}  %
\usepackage{helvet}  %
\usepackage{courier}  %
\usepackage[hyphens]{url}  %
\usepackage{graphicx} %
\usepackage[square,numbers]{natbib}  %
\usepackage{caption} %
\DeclareCaptionStyle{ruled}{labelfont=normalfont,labelsep=colon,strut=off} %
\usepackage{hyperref}
\usepackage{verbatim}
\usepackage{booktabs}
\usepackage{algorithm}
\usepackage{algorithmic}

\usepackage{newfloat}
\usepackage{listings}
\floatstyle{ruled}
\newfloat{listing}{tb}{lst}{}
\floatname{listing}{Listing}

\title{An Exact Algorithm for the Linear Tape Scheduling Problem}
\author{
    Valentin Honor\'e, \textsuperscript{\hskip-1.5mm\rm 1}
    Bertrand Simon, \textsuperscript{\hskip-1.5mm\rm 1}
    Fr\'ed\'eric Suter\textsuperscript{\rm 1,\rm 2} 
}
\date{
\small{

  \textsuperscript{\rm 1} IN2P3 Computing Center / CNRS, Lyon - Villeurbanne, France\\
  \textsuperscript{\rm 2} Oak Ridge National Laboratory, Oak Ridge, TN 37830, USA\\
    valentin.honore@cc.in2p3.fr, bertrand.simon@cc.in2p3.fr, frederic.suter@cc.in2p3.fr
}
}

\usepackage{xspace}
\usepackage{tikz}
\usepackage{amsthm,amsmath,amsfonts}
\usepackage{cleveref}
\newtheorem{definition}{Definition}
\newtheorem{theorem}{Theorem}
\newtheorem{lemma}{Lemma}

\newcommand{\virtuallb}{\textit{VirtualLB}\xspace}

\newcommand{\skipb}{\ensuremath{\textit{skip}(a,b,\nskip)}\xspace}
\newcommand{\detourc}{\ensuremath{\mathit{detour_c}(a,b,\nskip)}\xspace}
\newcommand{\uturn}{\ensuremath{U}\xspace}
\newcommand{\nfiles}{\ensuremath{n_{\mathit{req}}}\xspace}

\newcommand{\algo}[1]{\textsc{\bfseries #1}\xspace}
\newcommand{\DP}{\algo{DP}}
\newcommand{\simpledp}{\algo{SimpleDP}}
\newcommand{\logdp}{\algo{LogDP}}
\newcommand{\logdpone}{\algo{LogDP(1)}}
\newcommand{\logdpfive}{\algo{LogDP(5)}}
\newcommand{\nodetour}{\algo{NoDetour}}
\newcommand{\gs}{\algo{GS}}
\newcommand{\fgs}{\algo{FGS}}
\newcommand{\nfgs}{\algo{NFGS}}
\newcommand{\lognfgs}{\algo{LogNFGS}}

\newcommand{\opt}{\ensuremath{\mathit{OPT}}\xspace}

\newcommand{\ltsp}{\textsc{LTSP}\xspace}
\newcommand{\ie}{{\it i.e.,}\xspace}

\newcommand{\myxscale}{0.15cm}
\newcommand{\myyscale}{0.008cm}

\usepackage{transparent}

\usepackage{todonotes}

\newcommand{\CC}{IN2P3 Computing Center\xspace}

\makeatletter
\newcommand{\uturnpenalty}{15}
\newcounter{mytime}

\tikzset{myline/.style={thick, red!50!black}}%
\newcommand{\drawLine}[1]{
\setcounter{mytime}{0}
\ADD{100}{-#1}{\sol}
\ABSVALUE{\sol}{\sol}
\addtocounter{mytime}{\sol}
\draw[myline, |-] (100,0) --  node [at start, right, yshift=-20, align=left] {Reading\\ Head} 
              (#1,-\themytime);
\checknextarg{#1}
}
\newcommand{\checknextarg}[1]{\@ifnextchar\bgroup{\nextline{#1}}{}}
	\newcommand{\nextline}[2]{

	\ADD{#2}{-#1}{\sol}
	\ABSVALUE{\sol}{\sol}
	\ADD{\sol}{\uturnpenalty}{\sol}
	\addtocounter{mytime}{\sol}

	\@ifnextchar\bgroup{
		\draw[myline] (#1,-\themytime+\sol) --  (#1,-\themytime+\sol-\uturnpenalty) --  (#2,-\themytime);
		\nextline{#2}}{
	\draw[myline,->] (#1,-\themytime+\sol) --  (#1,-\themytime+\sol-\uturnpenalty)  --  (#2,-\themytime);}
}

\makeatother

\usepackage{calculator}
\begin{document}

\maketitle

\bigskip

\begin{abstract}
  Magnetic tapes are often considered as an outdated storage technology, yet
  they are still used to store huge amounts of data. Their main interests are a
  large capacity and a low price per gigabyte, which come at the cost of a much
  larger file access time than on disks. With tapes, finding the right ordering
  of multiple file accesses is thus key to performance. Moving the reading head
  back and forth along a kilometer long tape has a non-negligible cost and
  unnecessary movements thus have to be avoided. However, the optimization of
  tape request ordering has rarely been studied in the scheduling
  literature, much less than I/O scheduling on disks. For instance, 
  minimizing the average service time for several read requests on a linear
  tape remains an open question.

  Therefore, in this paper, we aim at improving the quality of service
  experienced by users of tape storage systems, and not only the peak
  performance of such systems. To this end, we propose a 
  reasonable polynomial-time exact algorithm while this problem and simpler
  variants have been conjectured NP-hard. We also refine the proposed model by
  considering U-turn penalty costs accounting for inherent mechanical
  accelerations. Then, we propose low-cost variants of our optimal algorithm by
  restricting the solution space, yet still yielding an accurate suboptimal
  solution. Finally, we compare our algorithms to existing  
  solutions from the literature on logs of the mass storage management system
  of a major datacenter. This allows us to assess the quality of
  previous solutions and the improvement achieved by our low-cost
  algorithm. Aiming for reproducibility, we make available the complete
  implementation of the algorithms used in our evaluation, alongside the
  dataset of tape requests that is, to the best of our knowledge, the first of
  its kind to be publicly released.
\end{abstract}

\bigskip

\section{Introduction}
Initially designed for media recording, the usage domain of
magnetic tapes has broadened over the decades and remains a real
competitor to disk storage even for scientific data. The main
advantages of this storage medium are a large storage capacity
for a reasonable price, a better data preservation, better
security, and better energy efficiency. Indeed, it has been
estimated that total costs are reduced by an average factor of 6 
when archiving data on tape rather than
disks~\cite{reine2015continuing}.

Recent tape cartridges can store up to 20 terabytes of data on a
one-kilometer-long physical storage, longitudinally divided into
few bands which are each also longitudinally divided into dozens
of wraps. Wraps are in turn divided into dozens of tracks. All
tracks in a given wrap are read or written simultaneously. A tape
is then composed of hundreds of parallel wraps which are
logically linked together in a \emph{linear
serpentine}. Intuitively, the storage space can be seen as a
single linear wrap coiled liked a serpent on the tape.

Thousands of such cartridges are usually stored on the shelves of robotic
libraries, as books would be stored in an actual library. Then, when data on a
given cartridge is not needed, its storage does not induce any power
consumption, and it cannot be accessed by intruders. All these advantages of
tape storage made it an unavoidable candidate for the storage of the exabytes
of data produced at CERN by the Large Hadron Collider
experiments~\cite{davis2019cern} or data related to European weather
forecast~\cite{ECMWF}.

The huge amount of data stored in such tape libraries, typically hundreds of
petabytes, is usually managed by a Mass Storage Management System ({\it e.g.,}
IBM HPSS or HPE DMF) which keeps track of the exact location of the files
stored on tapes and answers to users' requests. When a particular file is
needed, the tape it is on will be fetched by a robotic arm, brought to a tape
drive, and loaded. Then, the reading head of the tape drive is positioned to
the beginning of the file to read, or to the first available space to write new
data, and the I/O operation eventually occurs.

The main drawback of tape storage is the high latency to access a given
file. Mounting a tape into a tape reader requires a delay of about a
minute~\cite{cano2021cern}. Moreover, seeking from one file to another adds
more delay to place the reading head on the correct wrap and adapt the
longitudinal position of the tape in front of the head.  When accesses to
multiple files are requested, finding the right ordering of these accesses is
thus key to performance. Moving the reading head back and forth along a
kilometer long tape has a non-negligible cost and unnecessary movements thus
have to be avoided. However, the optimization of tape request ordering has
rarely been studied in the scheduling literature, much less than I/O scheduling
on disks. For instance, minimizing the average service time for several
read requests, \ie the average time at which each request is read, on a linear
tape remains an open question.

Therefore, in this paper, we aim at improving the quality of service
experienced by users of tape storage systems, and not only the peak performance
of such systems. To this end, we consider a simplified model of magnetic tape
composed of a single linear track. This is a strong assumption as the
serpentine nature of tapes leads to important optimization decisions. However,
it still reflects local batch requests which would target files belonging to
the same wrap. We also believe it is a fundamental model which should be deeply
understood. In this model, a tape can therefore be seen as a linear sequence of
files which all have to be read from the left to the right. The input of the
problem we consider is a list of files that are requested, associated with a
number of requests for each file. The objective is to design a schedule (\ie a
trajectory of the reading head on the linear tape) to read all the requested
files when the reading head is initially positioned on the right of the
tape. We consider the average service time as a metric, to ensure a fair
service among all requests. In order to model the temporality of a given
schedule, we assume that the speed of the tape movement is constant, although it is a
mechanical device with inertia. We moderate this inaccuracy by taking into
account the deceleration induced by a U-turn of the tape as a nominal penalty.
Note that we do not consider write requests, which are usually done separately,
nor update requests, which are avoided as they damage nearby data.
Following~\cite{cardonha2016}, we refer to this problem as the Linear Tape
Scheduling Problem~(\ltsp), noting that our model differs from theirs by
accounting for U-turn penalties.

\ltsp has been previously studied by~\citet{cardonha2016,cardonha2018} and
conjectured to be impossible to be solved efficiently. Indeed, even simpler
variations restricting either file requests to be unique or file sizes to be
equal have been conjectured NP-hard~\cite{cardonha2018}. We answer this open
question in this paper by providing a polynomial algorithm optimally solving
the unrestricted \ltsp problem, also considering U-turn penalties. More
precisely, we show that a carefully designed Dynamic Programming implementation
(technique which has been considered in~\cite{cardonha2018} but was deemed not
conclusive) allows us to compute an optimal schedule in a reasonable polynomial
time. We then provide faster suboptimal algorithms and compare the performance
of these original algorithms to that of existing algorithms on a dataset
built from the recent history of the tape library of the \CC.~\footnote{We discuss the connections with a concurrent work \cite{cardonha21} in \Cref{app:cardonha21}.}

The remainder of this paper is organized as follows. In \Cref{sec:related}, we
review the literature on tape scheduling and related optimization problems. In
\Cref{sec:framework}, we define and discuss precisely the model and the
objective function. In \Cref{sec:algo} we expose our algorithmic solutions to
this problem. Finally, in \Cref{sec:expe}, we present the results of our
simulations on a real-world dataset.

\section{Related work}
\label{sec:related}

The closest works to the present paper~\cite{cardonha2016,cardonha2018} study
\ltsp under the same tape model, but without U-turn penalties. The authors note
that the algorithm minimizing the maximal service time, \ie the time at which
all files are read, can present an average service time arbitrarily far from
the optimal. They show that the opposite algorithm reading the rightmost files
first is however a 3-approximation, and design a few greedy optimizations.
Finally, they provide several heuristics for the online variant and
compare their solutions through simulations.

\ltsp is related to several well-studied problems in theoretical computer
science. The most famous is probably the Traveling Salesperson Problem, where
the goal is to visit $n$ points as fast as possible following given travel
times between each pair of points. This problem is notoriously
NP-hard in general metrics~\cite{lawler1985} so approximation algorithms and
special cases have been studied extensively. One of the most recent development
has been the design of an algorithm surpassing the long-standing approximation
ratio of $1.5$~\cite{karlin2021slightly}. \ltsp is closer to its restriction on
the real line, for which it can be solved in $O(n^2)$~\cite{bjelde2020}.

A key difference between \ltsp and the Traveling Salesperson Problem resides in
the objective function, as \ltsp aims at minimizing the average service time.
This objective is captured by the Minimum Latency Problem (also called
Traveling Repairperson Problem) for which the best known approximation ratio is
$3.59$~\cite{chaudhuri2003}. This problem is already strongly NP-hard on
trees~\cite{sitters2002minimum}, although it admits a
PTAS~\cite{sitters2021polynomial}, but can be solved polynomially on the line
if there are no deadlines~\cite{afrati1986complexity}.

Keeping the average service time objective function but adding delays at every
visited vertex leads to a more general definition of the Traveling Repairperson
Problem. This problem is strongly NP-hard on the line when deadlines or release
times are involved~\cite{BOCK2015690} but its complexity when requests can be
served at any time is still unknown.

A different kind of related problems has been studied under the name of
Dial-A-Ride.  Here, requests are composed of a source and a destination and the
goal is to move vehicles to transport all requests from their source to their
destination.  Several variants of the problem exist, even restricted to the
offline setting, depending on the presence of release times or the number and
capacities of the vehicles, see~\cite{de2004computer}. The Dial-A-Ride problem
can be seen as a generalization of \ltsp but is often studied with the
objective of minimizing the total service time. A simpler variant, close to
our problem, considers a single vehicle able to transport one request at a time
without being able to drop it before the destination, and is shown to be
polynomially solvable~\cite{atallah1988efficient} when minimizing the total
service time.  A formulation aiming at minimizing the average service
time has been shown to be NP-hard, relying on request irregularities
(overlapping trips in different directions)~\cite[Theorem 7]{de2004computer}
which cannot happen in \ltsp where requests are unidirectional and files are
disjoint.

We did not cover all the work done on the online version of these problems,
when future requests are unknown, but we refer the reader to~\cite{bjelde2020}
for an overview of such results.

The literature on tape scheduling is rather scarce although the role of tape
libraries is far from negligible in modern computing centers. Contrarily to
this paper, most studies consider a more complex tape geometry, usually a
serpentine. \citet{hillyer1996modeling} focus on low-level hardware information
(\emph{key points}) to evaluate several
heuristics. \citet{sandsta1999improving} propose a low-cost function to
approximate the seeking time between two points of the tape.
\citet{more2000scheduling} design algorithms to schedule the mounts of
different tapes in a library. \citet{talkCERN} evaluates the seek times between
any two points of a recent tape, data which is used as input in a few
heuristics to compare their performance. Software designed to optimize tape
usage appear to often sort read requests based on their tape
position~\cite{schaeffer2011treqs,zhang2006hptfs}. A common point to
these studies is that the focus has mostly been on cost modeling due to the
two-dimensional nature of the tape and low-level hardware aspects, but publicly
released scheduling algorithms are often greedy ones. A proprietary solution
used by some tape libraries, named Recommended Access Order~(RAO), exploits
such two-dimensional tape information but its underlying algorithm is not
available~\cite[Section 4.27]{IBMmanual}.

\section{Model and Problem Descriptions}
\label{sec:framework}

We consider a linear tape of length $m$, divided successively in $n_f$ disjoint
files $(f_1,\dots,f_{n_f})$ of integer size $s(f_i)$. Let $\ell(f_i)$ be the
\emph{length} between the left of the tape and the left of the file $f_i$ and
$r(f_i)=\ell(f_i)+s(f_i)$. We say that $f_i<f_j$ if file $f_i$ is located on the
left of $f_j$, \ie $\ell(f_i)<\ell(f_j)$. We assume that these file
properties can be queried in constant time by an algorithm. We are given a set
of $n$ requests on $\nfiles$ files among the $n_f$ files of the tape, with
possible duplicates, where each request is a file. Let $x(f_i)$ be the number of
requests allocated to file $f_i$.

At the beginning, the reading head is positioned on the right of the tape. A
request is fulfilled when its file has been traversed from the left to the right
by the reading head. We assume the reading head moves at constant speed (the
tape is actually moving and the head is fixed, but switching roles helps the
exposure), a time unit being necessary to traverse a file chunk of size 1 in
either direction. We also consider a time penalty \uturn for each U-turn
performed by the head.

The main limitation of this model concerns the track geometry. Modern tapes are
not constituted of a single linear track, and being aware of their serpentine
geometry is essential to optimize the reading sequence and seeking costs.
However, this simpler model is able to emulate accurately local considerations
when files written in the same period are located in a single track. It is also
fundamental to deeply understand the complexity of such a model knowing that
the serpentine model is much closely related to NP-hard problems such as the
Traveling Salesperson Problem.

The assumption of the tape moving at a constant speed in front of the reading
head is obviously inaccurate due to acceleration and deceleration inherent to
mechanical devices. However, the cruise speed is typically reached fast enough so this
approximation is  satisfactory apart from U-turns. The nominal U-turn
penalty used to take into account these slow-downs therefore improves the model
accuracy.

Other limitations of the model such as the undifferentiated reading speed
or the forced starting position of the head are discussed as extensions
in the conclusion.

The objective is to provide a {\it schedule}, \ie a trajectory of the reading
head on the tape, that serves all requests and minimizes the sum of service
times of requests, \ie the sum of the times needed before each request is
satisfied. Note that we formally define the objective as minimizing the sum,
but it is more intuitive in terms of a quality of service to speak about the
average service time, an objective which is completely equivalent.

A simple lower bound \virtuallb on the optimal solution is achieved by using
$n$ virtual heads serving each request optimally, \ie each reading head
moves directly to the left of its assigned file then reads it.

$$ \virtuallb = \sum_{f} x(f)\cdot (m-\ell(f) + s(f) + \uturn ) .$$

Minimizing the average service time is one of the most classic scheduling
objective functions with the maximal service time. The latter has been the main
focus of studies on the serpentine model as it minimizes the time spent using
the tape which decreases wear and delay of other tapes reads. However, in the
linear tape model, minimizing the maximal service time is trivial while
minimizing the average service time leads to more fairness among users. This is
especially true in a case of low tape usage in which tapes are rarely waiting
to be mounted.

Note that we follow the definition of the problem
from~\cite{cardonha2016,cardonha2018} as the input consists of a list of
requests rather than the set of requested files associated with their
multiplicity. The motivation comes from practice, where a set of read requests
has to be satisfied, and it may happen that several read requests target the
same file. The consequence is that polynomial-time algorithms are allowed to
have a complexity polynomial in $n$ and $\nfiles$ and not only in $\log n$ and
$\nfiles$. This makes a difference if the number of requests is not bounded by a
polynomial in the number of requested files. It is natural to study first this variant of the problem, as so-called
high-multiplicity problems are notoriously much harder to solve~\cite{gabayphd}.

\section{Algorithm}
\label{sec:algo}

This section presents the main contribution of this paper, the \DP algorithm
solving \ltsp in time $O(\nfiles^3\cdot n)$.  Before describing \DP, we start
with giving useful definitions, preliminary remarks, and brief descriptions of
existing solutions. We then also present the \logdp variant algorithm, which
limits the search space of \DP to provide a suboptimal solution with a smaller
time complexity of $O(\nfiles\cdot n \cdot \log^2 \nfiles)$. 

\subsection{Preliminaries}

In this section, we study the structure of optimal solutions to provide a simple
description of such schedules.

In any optimal solution, the reading head will move to the leftmost request,
then to the rightmost still unread request. Before reaching the leftmost
request, the head may move back and forth in possibly intricate patterns to
read relevant files first. We say that the solution includes the \emph{detour}
$(a,b)$, with $a$ and $b$ being two requested files such that $a\leq b$, if the head goes
directly to $r(b)$ then back to $\ell(a)$ after first attaining $\ell(a)$. As shown
previously~\cite{cardonha2016} and later stated formally in our setting (see
\Cref{lem:laminar}), there always exists an optimal solution which can be
described only via a set of detours. Furthermore, a detour can be totally
surrounded by a later one (\ie $(a_1,b_1)$ and $(a_2,b_2)$ with
$a_1<a_2<b_2<b_1$) but otherwise two detours cannot intersect each other (\ie
$(a_1,b_1)$ and $(a_2,b_2)$ with $a_1\leq a_2\leq b_1\leq b_2$).

\newcommand{\addfile}[3][f_\thefile]{
\addtocounter{file}{1}
\fill[blue!20!white] (#2,45) rectangle (#3,-215-7*\uturnpenalty);
\node[blue!50!white] at ({(#3+#2)*0.5}, -255-7*\uturnpenalty) {$#1$};
\draw[blue!30!white, ultra thick, dashed] ({(#3+#2)*0.5}, 45) -- (file);
}
\begin{figure}[h!]
\centering
\resizebox{\linewidth}{!}{
\begin{tikzpicture}[x=\myxscale,y=\myyscale]
\newcounter{file}
\node (file) at (50,150) {Requested files};
\addfile{20}{25}
\addfile{30}{35}
\addfile{40}{50}
\addfile{60}{65}
\addfile{70}{80}
\addfile{90}{95}

\drawLine{90}{95}{60}{65}{40}{80}{20}{35}

\draw[ultra thick, |-|] (100,45) -- + (-90,0) node [at start, right] {Tape};

\draw[thick, ->] (8,-50) -- + (0,-80) node [midway, left] {\it time~~};

\end{tikzpicture}
}
\caption{Example of schedule for reading six files described by the
  $[(f_6,f_6),(f_4,f_4),(f_3,f_5)]$ detour list. Note the delays caused by U-turn penalties.}
\label{fig:detours}
\end{figure}
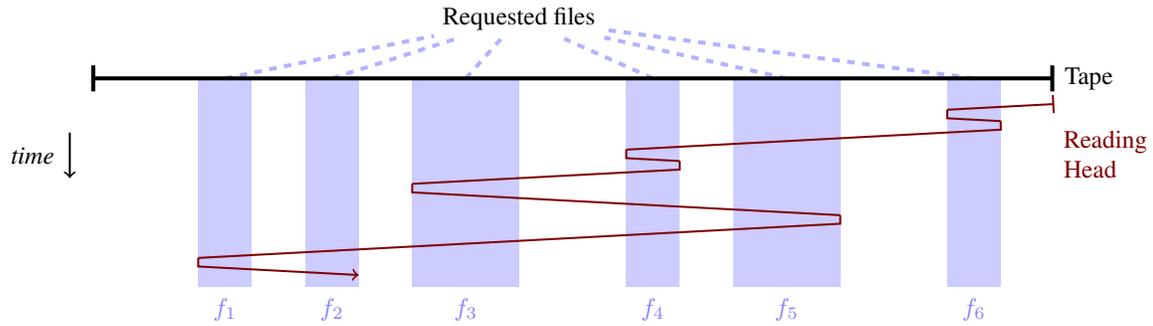

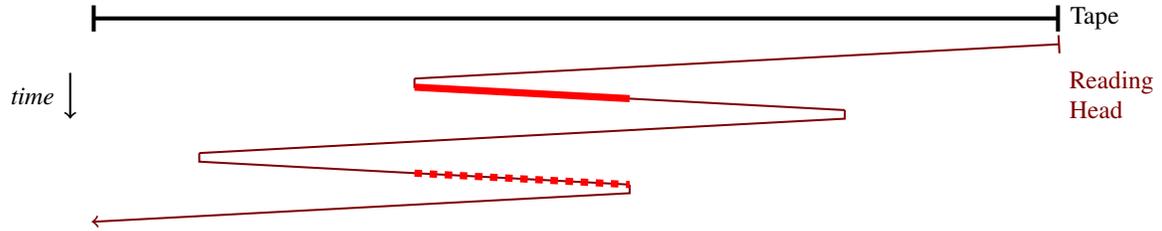
\begin{figure}[h!]
\centering
\resizebox{\linewidth}{!}{
\begin{tikzpicture}[x=\myxscale,y=\myyscale]
\drawLine{40}{80}{20}{60}{10}

\draw[line width=1mm, dashed, red] (40,-180-3*\uturnpenalty) -- (60,-200-3*\uturnpenalty);
\draw[line width=1mm, red] (40,-60-\uturnpenalty) -- (60,-80-\uturnpenalty);

\draw[ultra thick, |-|] (100,45) -- + (-90,0) node [at start, right] {Tape};

\draw[thick, ->] (8,-50) -- + (0,-80) node [midway, left] {\it time~~};

\end{tikzpicture}}
\caption{Example of non-optimal schedule. In the second detour, the movement in
  thick dotted lines is useless as these files have already been read earlier (thick solid line).}
\label{fig:nodetour}
\end{figure}

\Cref{fig:detours} illustrates a possible solution while
\Cref{fig:nodetour} shows detours overlapping in a suboptimal manner.

We denote this property on the set of detours in any optimal solution as being
\emph{strictly laminar}, following a definition of {\it laminar} used in the
scheduling literature, see for instance~\cite{chen2018}. We consider that all solutions
contain the detour $(f_{n_1},f_{n_f})$, which reads all skipped files, even if
the last movements may not count towards the objective as the rightmost
requests may have already been served.

An unread file at the right of the current reading head position is called
\emph{skipped}. It will be read later when the head moves back to the right,
possibly after the head read the leftmost file. For instance, on
\Cref{fig:detours}, when $f_4$ is first reached by the head, $f_5$ is skipped,
but when the head first reaches $f_2$, no file is skipped.

\subsection{Existing algorithms}
\label{sec:prelim}

One of the simplest algorithm would be to make no detour. The head simply moves
to the leftmost requested file and then reads all files left-to-right. Despite
minimizing the makespan, it can be arbitrarily far from the optimal solution in
our model~\cite{cardonha2016}. We refer to this algorithm as \nodetour.

The opposite strategy would be to perform a detour on each requested file. This
algorithm, named \gs for {\bfseries G}reedy {\bfseries S}cheduling, has been
proved to be a 3-approximation without U-turn penalties~\cite{cardonha2016}. The worst-case instance is simply composed of a small file with many requests located on the left of a large file with a single request.
But of course harsh penalties can arbitrarily degrade
its guarantees.

To improve the basic solution offered by \gs, the \fgs
algorithm~\cite{cardonha2018} detects detrimental detours in multiple
evaluation passes and {\bfseries F}ilters them out.

As \fgs does not benefit from multi-file detours, the same authors designed the
\nfgs algorithm, allowing {\bfseries N}on-atomic detours. In essence, for each pair of files $a<b$ starting from the
left, it tests whether it would be beneficial to add the detour $(a,b)$,
after removing the detour starting from $a$ if it existed. Despite its
relatively large time complexity,  \nfgs remains greedy in nature,
definitely sealing any detour that seems beneficial. A variant exploring only
detours spanning over $O(\log \nfiles)$ requested files, \lognfgs, has been proposed
to trade search space for running time.

Note that the \fgs, \nfgs, and \lognfgs algorithms can be adapted to take into
account the U-turn penalty in their decisions, although losing their
approximation factor of 3 which was inherited from \gs. We provide a description
in Appendix~\ref{app:algos} for completeness.

The structure of existing solutions, relying on greedy evaluation passes,
illustrates the difficulty of the problem. The decision of making a detour or
not depends on what happens before (a detour increases the delay on skipped
files) and after (subsequent detours will increase the delay on files that have
been skipped).  Detours can also be intricate, as shown
by~\Cref{fig:detours}. It thus seems hardly possible to take correct decisions
on detours when each decision may influence the others. Consequently,
\citet{cardonha2018} only considered a very restricted model (identical file
sizes and a single request per file) in which the exact solution is simple but
did not otherwise get any algorithm with an approximation ratio below 3.

\subsection{Algorithm}

\newcommand{\nskip}{\ensuremath{n_{\textit{skip}}}\xspace}
\newcommand{\fab}{\ensuremath{F_{a,b}}\xspace}
\newcommand{\nleft}[1]{\ensuremath{n_{\ell}(#1)\xspace}}
\newcommand{\leftb}[1]{\ensuremath{\textit{left}(#1)\xspace}}

Here, we describe the \DP dynamic programming algorithm. It uses carefully
selected memoization to store the cost of specific solutions used to build an
optimal schedule.

The dynamic program cells have a number and three parameters: two requested files
$a$ and $b$ and a number $\nskip<n$. The objective for each cell is to compute
the best possible strategy for the reading head between $r(b)$ and $\ell(a)$
knowing that:
\begin{enumerate} 
\item there is a detour $(a,f)$ for some file $f\geq b$,
\item there is no detour $(f_1,f_2)$ for any files $f_1$, $f_2$ satisfying $a<f_1<b<f_2$,
\item when the reading head reaches $r(b)$,
exactly $\nskip$ file requests have been skipped.
\end{enumerate}
 The content of the cell describes the impact on the total cost of the movement
 made by the reading head between the first time it reaches $r(b)$ and the
 first time it reaches $r(b)$ after having read $a$. In other words, it equals
 the sum of the lengths for all requests on any file $f$ of the ``unnecessary''
 paths traversed by the head in this time interval and before serving the file
 $f$. Unnecessary means that we do not count the cost that would also be
 incurred to \virtuallb on a file $f$ between $a$ and $b$, as it is inevitable
 and this simplifies the formulas. The U-turn penalty on $a$ is therefore not
 counted as \virtuallb would also have one U-turn penalty, but other U-turn
 penalties in this interval are counted.

We define \nleft{b} as the number of requests on files located on the left of
$b$, excluding $b$, and let \leftb{b} be the closest requested file located to
the left of $b$.

The value of cell $T[a,b,\nskip]$ is then defined as follows:
\begin{itemize}
\item If $b=a$, then there is a detour from $\ell(b)$ to at least $r(b)$ so we delay all
pending requests by $2s(b)$, and incur no additional cost to $b$, see
\Cref{fig:detourfb,fig:detourbf}. Therefore,
$$T[b,b,\nskip] =  2 \cdot s(b) \cdot
(\nskip + \nleft{b}).$$

\renewcommand{\addfile}[3][f_\thefile]{
\addtocounter{file}{1}
\fill[blue!20!white] (#2,45) rectangle (#3,-145-3*\uturnpenalty);
\node[blue!50!white] at ({(#3+#2)*0.5}, -175-3*\uturnpenalty) {$#1$};
}

\begin{figure}[tb]
\centering
\resizebox{.91\linewidth}{!}{
\begin{tikzpicture}[x=\myxscale,y=\myyscale*1.5]

\addfile[b]{40}{55}

\addfile[f]{60}{70}

\draw[thick, |-, red!50!black] (100,0) --  node [at start, right, yshift=-20, align=left] {Reading\\ Head} 
                  (60,-40) ;
\draw[thick, ->, red!50!black] (60,-80-\uturnpenalty) -- (70,-90-\uturnpenalty);

\draw[ultra thick, dashed, red] (55,-45) -- (40,-60) -- (40,-60-\uturnpenalty) -- (55,-75-\uturnpenalty);

\draw[ultra thick, |-|] (100,45) -- + (-90,0) node [at start, right] {Tape};

\draw[thick, ->] (8,-50) -- + (0,-80) node [midway, left] {\it time~~};

\end{tikzpicture}}
\caption{Cost incurred by a detour over file $b$ to a skipped file $f$. The
  solid line represents the shortest path to serve $f$. The red dotted line
  represents the delay incurred by this detour to the service time of
  $f$. Other detours are not illustrated here. Subsequent figures follow the
  same logic.}
\label{fig:detourbf}
\end{figure}
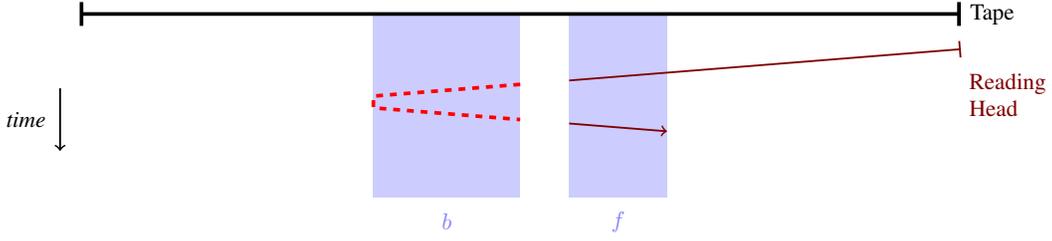

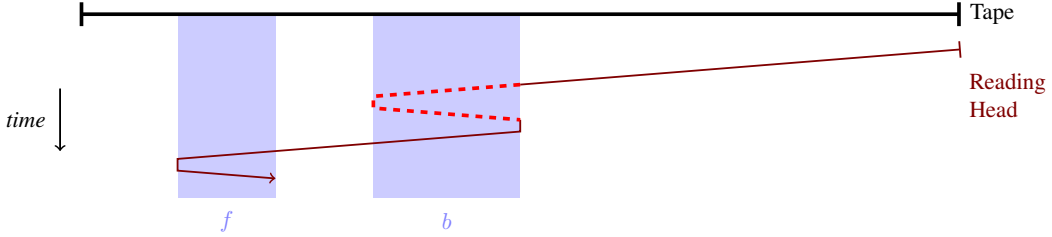
\begin{figure}[tb]
\centering
\resizebox{.91\linewidth}{!}{
\begin{tikzpicture}[x=\myxscale,y=\myyscale*1.5]

\addfile[b]{40}{55}

\addfile[f]{20}{30}

\draw[thick, |-, red!50!black] (100,0) --  node [at start, right, yshift=-20, align=left] {Reading\\ Head} 
                  (55,-45) ;
\draw[thick, ->, red!50!black] (55,-75-\uturnpenalty) -- (55,-75-2*\uturnpenalty) -- (20,-110-2*\uturnpenalty) -- (20,-110-3*\uturnpenalty) -- (30,-120-3*\uturnpenalty);

\draw[ultra thick, dashed, red] (55,-45) -- (40,-60) -- (40,-60-\uturnpenalty) -- (55,-75-\uturnpenalty);

\draw[ultra thick, |-|] (100,45) -- + (-90,0) node [at start, right] {Tape};

\draw[thick, ->] (8,-50) -- + (0,-80) node [midway, left] {\it time~~};

\end{tikzpicture}}
\caption{Cost incurred by the detour over $b$ to a left file $f$.}
\label{fig:detourfb}
\end{figure}

\item Otherwise, let \fab be the set of requested files located between
$a$ and $b$ excluding $a$. There are several possibilities to consider to determine
the value of the cell: either $b$ is skipped (it will be read with the detour
starting from $a$), or read sooner than by the detour starting from $a$.
In the latter case, it is read on a detour ending on $b$ as there is no detour
going to the right of $b$ starting righter than~$a$. This detour can start
from any file in \fab. Then, we have:

\begin{align*}
&\skipb := ~  T[a,\leftb{b},\nskip+x(b)] \\
&~~~~+ 2\cdot (r(b)-r(\leftb{b}))\cdot (\nskip+\nleft{a})\\
&~~~~+ 2\cdot (\ell(b)-r(\leftb{b}))\cdot x(b)\\\\
&\detourc := ~  T[a,\leftb{c},\nskip]+T[c,b,\nskip]\\
&~~~~+ 2\cdot (r(b)-r(\leftb c))\cdot (\nskip+\nleft{a})\\
&~~~~+ 2\cdot U\cdot (\nskip+\nleft{c})
\\
\\
&T[a,b,\nskip] = ~   \min\big( ~ \textit{skip}(a,b,\nskip) ~;\\ 
&~~~~\hspace{2.68cm}
\min_{c\in \fab} ~  \mathit{detour_c}(a,b,\nskip)~ \big)
\end{align*}

In the first case, we recurse on a smaller window skipping  file $b$, hence
increasing \nskip. We also account for the cost of the detour starting from $a$
over the files between \leftb b and $b$ for the requests that will be fulfilled
later. The differences with earlier are that (1) we also have to account for
the cost to traverse the unrequested files at the left of $b$ and (2) requests
between $a$ and $\leftb b$ are served before the head comes back to the right,
hence there are \nleft{a} delayed files and not \nleft{b}. See
\Cref{fig:detourskip,fig:detourskipleft}.
\renewcommand{\addfile}[3][f_\thefile]{
\addtocounter{file}{1}
\fill[blue!20!white] (#2,45) rectangle (#3,-215);
\node[blue!50!white] at ({(#3+#2)*0.5}, -255) {$#1$};
}
\begin{figure}[h!]
\centering
\resizebox{.91\linewidth}{!}{
\begin{tikzpicture}[x=\myxscale,y=\myyscale]
\addfile[a]{20}{30}
\addfile[\leftb b]{40}{50}
\addfile[b]{60}{70}
\addfile[f]{90}{95}

\draw[thick, |-, red!50!black] (100,0) --  node [at start, right, yshift=-20, align=left] {Reading\\ Head} 
                  (90,-10) ;
\draw[thick, ->, red!50!black] (90,-140-\uturnpenalty) -- (95,-145-\uturnpenalty);

\draw[ultra thick, dashed, red] (70,-30) -- (50,-50);
\draw[ultra thick, dashed, red] (50,-100-\uturnpenalty) -- (70,-120-\uturnpenalty);

\draw[thick, dotted, black] (50,-50) -- (20,-70) -- (20,-70-\uturnpenalty) -- (50,-100-\uturnpenalty);

\draw[ultra thick, |-|] (100,45) -- + (-90,0) node [at start, right] {Tape};

\draw[thick, ->] (8,-50) -- + (0,-80) node [midway, left] {\it time~~};

\end{tikzpicture}}
\caption{Impact of $skip(a,b,\nskip)$ on a skipped file
  $f$. The thin dotted line represents the recursively computed impact (which
  may include subsequent detours), and the dashed line the impact directly
  accounted for.
}
\label{fig:detourskip}
\end{figure}
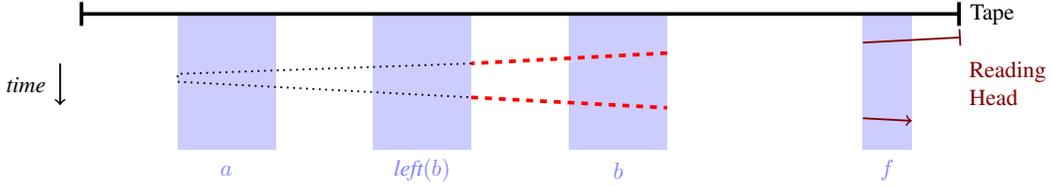
\renewcommand{\addfile}[3][f_\thefile]{
\addtocounter{file}{1}
\fill[blue!20!white] (#2,45) rectangle (#3,-215-3*\uturnpenalty);
\node[blue!50!white] at ({(#3+#2)*0.5}, -255-3*\uturnpenalty) {$#1$};
}

\begin{figure}[h!]
\centering
\resizebox{.91\linewidth}{!}{
\begin{tikzpicture}[x=\myxscale,y=\myyscale]
\addfile[a]{40}{50}
\addfile[\leftb b]{60}{70}
\addfile[b]{80}{90}
\addfile[f]{20}{30}

\draw[thick, |-, red!50!black] (100,0) --  node [at start, right, yshift=-20, align=left] {Reading\\ Head} 
                  (90,-10) ;
\draw[thick, ->, red!50!black] (90,-110-\uturnpenalty) -- (90,-110-2*\uturnpenalty) -- (20,-180-2*\uturnpenalty) -- (20,-180-3*\uturnpenalty) -- (30,-190-3*\uturnpenalty);

\draw[ultra thick, dashed, red] (70,-90-\uturnpenalty) -- (90,-110-\uturnpenalty) (70,-30) -- (90,-10);

\draw[thick, dotted, black] (70,-30) -- (40,-60)  -- (40,-60-\uturnpenalty) -- (70,-90-\uturnpenalty) ;

\draw[ultra thick, |-|] (100,45) -- + (-90,0) node [at start, right] {Tape};

\draw[thick, ->] (8,-50) -- + (0,-80) node [midway, left] {\it time~~};

\end{tikzpicture}}
\caption{Illustration of the impact of $skip(a,b,\nskip)$ on a left file $f$. 
}
\label{fig:detourskipleft}
\end{figure}
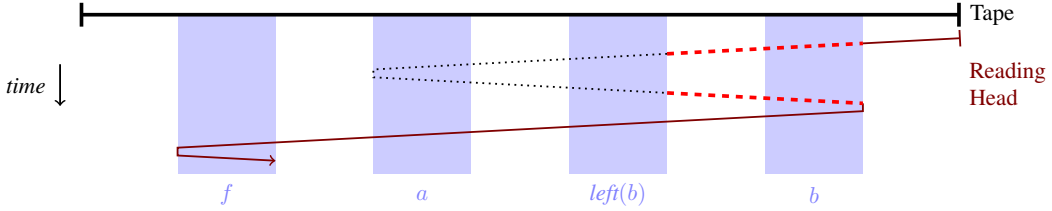%

Finally, we account for the additional cost to serve $b$ not covered by the
recursive call: the path over the unrequested files directly at the left of
$b$, see \Cref{fig:detourskipb}.

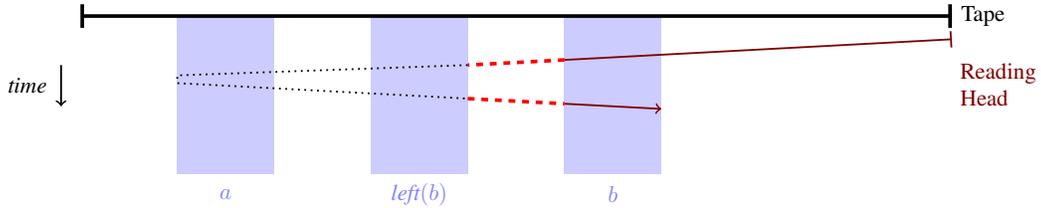
\begin{figure}[h!]
\centering
\resizebox{.9\linewidth}{!}{
\begin{tikzpicture}[x=\myxscale,y=\myyscale]

\addfile[a]{20}{30}
\addfile[\leftb b]{40}{50}
\addfile[b]{60}{70}

\draw[thick, |-, red!50!black] (100,0) --  node [at start, right, yshift=-20, align=left] {Reading\\ Head} 
                  (60,-40) ;
\draw[thick, ->, red!50!black] (60,-110-\uturnpenalty) -- (70,-120-\uturnpenalty);

\draw[ultra thick, dashed, red] (60,-40) -- (50,-50);
\draw[ultra thick, dashed, red] (50,-100-\uturnpenalty) -- (60,-110-\uturnpenalty);

\draw[thick, dotted, black] (50,-50) -- (20,-70) -- (20,-70-\uturnpenalty) -- (50,-100-\uturnpenalty);

\draw[ultra thick, |-|] (100,45) -- + (-90,0) node [at start, right] {Tape};

\draw[thick, ->] (8,-50) -- + (0,-80) node [midway, left] {\it time~~};

\end{tikzpicture}}
\caption{Impact of $skip(a,b,\nskip)$ on $b$.
}
\label{fig:detourskipb}
\end{figure}

In the second case, we have a detour $(c,b)$ for some $c$ in \fab. Hence, all
these files will be read when the head reaches $\leftb c$ so we do not change
\nskip in the recursive calls. We still need to account for the cost of the
detour starting from $a$ over the interval $(r(\leftb c),b)$. See
\Cref{fig:detourunskip,fig:detourunskipleft}. We also charge here the U-turn
penalties for all requests who will be served after the head reaches $a$, \ie
for all pending requests for which the U-turn at $c$ is not the last one before
they get served (the second U-turn penalty charged is for the U-turn occurring
at $b$ after the detour $(c,b)$).

\end{itemize}

Then, the overall solution can be computed through the call to
$T[f_1,f_{n_f},0]$. The structure of the recursive calls minimizing this value
leads to the detours taken by the underlying optimal solution.

\subsection{Proof of the algorithm}

First, we need a structural result to guarantee that the restriction to strictly laminar detours preserves the optimal solution. A similar result has been established in \cite{cardonha2016}. We state it here for self-consistency and a more precise result.

\renewcommand{\addfile}[3][f_\thefile]{
\addtocounter{file}{1}
\fill[blue!20!white] (#2,45) rectangle (#3,-235-3*\uturnpenalty);
\node[blue!50!white] at ({(#3+#2)*0.5}, -275-3*\uturnpenalty) {$#1$};
}
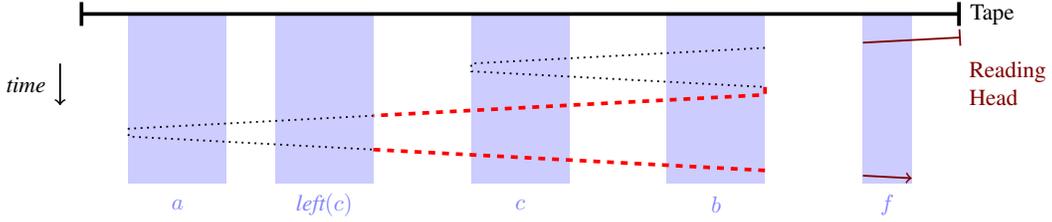
\begin{figure}[h!]
\centering
\resizebox{.91\linewidth}{!}{
\begin{tikzpicture}[x=\myxscale,y=\myyscale]
\addfile[a]{15}{25}
\addfile[\leftb c]{30}{40}
\addfile[c]{50}{60}
\addfile[b]{70}{80}
\addfile[f]{90}{95}

\draw[thick, |-, red!50!black] (100,0) --  node [at start, right, yshift=-20, align=left] {Reading\\ Head} 
                  (90,-10) ;
\draw[thick, ->, red!50!black] (90,-220-3*\uturnpenalty) -- (95,-225-3*\uturnpenalty);

\draw[ultra thick, dashed, red] (80,-80-\uturnpenalty) -- (80,-80-2*\uturnpenalty) -- (40,-120-2*\uturnpenalty);
\draw[ultra thick, dashed, red] (40,-170-3*\uturnpenalty) -- (80,-210-3*\uturnpenalty);

\draw[thick, dotted, black] (80,-20) -- (50,-50) -- (50,-50-\uturnpenalty) -- (80,-80-\uturnpenalty) ;

\draw[thick, dotted, black] (40,-120-2*\uturnpenalty) -- (15,-145-2*\uturnpenalty) -- (15,-145-3*\uturnpenalty) -- (40,-170-3*\uturnpenalty);

\draw[ultra thick, |-|] (100,45) -- + (-90,0) node [at start, right] {Tape};

\draw[thick, ->] (8,-50) -- + (0,-80) node [midway, left] {\it time~~};

\end{tikzpicture}}
\caption{Impact of $\mathit{detour_c}(a,b,\nskip)$ on a skipped file $f$. 
}
\label{fig:detourunskip}
\end{figure}

\renewcommand{\addfile}[3][f_\thefile]{
\addtocounter{file}{1}
\fill[blue!20!white] (#2,45) rectangle (#3,-315-5*\uturnpenalty);
\node[blue!50!white] at ({(#3+#2)*0.5}, -355-5*\uturnpenalty) {$#1$};
}
\begin{figure}[h!]
\centering
\resizebox{.91\linewidth}{!}{
\begin{tikzpicture}[x=\myxscale,y=\myyscale]
\begin{scope}[xshift=10*0.15cm]
\addfile[a]{15}{25}
\addfile[\leftb c]{30}{40}
\addfile[c]{50}{60}
\addfile[b]{70}{80}
\addfile[f]{5}{10}

\draw[thick, |-, red!50!black] (86,-14) --  node [at start, right, yshift=-20, align=left] {Reading\\ Head} 
                  (80,-20) ;
\draw[thick, ->, red!50!black] (80,-210-3*\uturnpenalty) -- (80,-210-4*\uturnpenalty) -- (5,-285-4*\uturnpenalty) -- (5,-285-5*\uturnpenalty) -- (10, -290-5*\uturnpenalty);

\draw[ultra thick, dashed, red] (80,-80-\uturnpenalty) -- (80,-80-2*\uturnpenalty) -- (40,-120-2*\uturnpenalty);
\draw[ultra thick, dashed, red] (40,-170-3*\uturnpenalty) -- (80,-210-3*\uturnpenalty);

\draw[thick, dotted, black] (80,-20) -- (50,-50) -- (50,-50-\uturnpenalty) -- (80,-80-\uturnpenalty) ;

\draw[thick, dotted, black] (40,-120-2*\uturnpenalty) -- (15,-145-2*\uturnpenalty) -- (15,-145-3*\uturnpenalty) -- (40,-170-3*\uturnpenalty);
\end{scope}

\draw[ultra thick, |-|] (100,45) -- + (-90,0) node [at start, right] {Tape};

\draw[thick, ->] (8,-50) -- + (0,-80) node [midway, left] {\it time~~};

\end{tikzpicture}}

\caption{Impact of $\mathit{detour_c}(a,b,\nskip)$ on a left file $f$. 
}
\label{fig:detourunskipleft}
\end{figure}
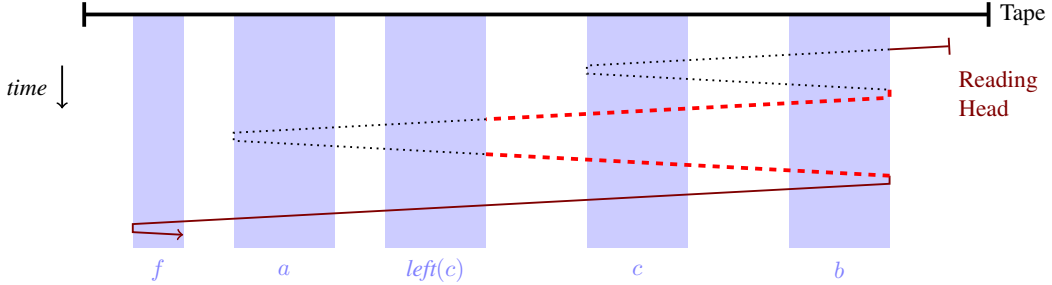

\begin{lemma}
\label{lem:laminar}
There exists an optimal solution composed only of strictly laminar detours.%
\end{lemma}

\begin{proof}
Consider an optimal
  solution. Once the leftmost file is reached, it must go straight to the
  rightmost unread file. We now consider the part of solution before the
  leftmost file is reached.

  Each time the head turns to the right at position $x$, it has to turn back to
  the left later at point $y$. It cannot turn again to the right before
  reaching $x$ as this is suboptimal: no new file between $x$ and $y$ can be
  read this way. Furthermore, $x$ must be the left of a requested file $a$ and
  $y$ the right of a requested file $b$ or this is suboptimal. So the solution
  made a detour $(a,b)$. Continuing this analysis, we can decompose the optimal
  solution as a set of detours, counting again a global detour
  $(f_1,f_{n_f})$. Note that we have shown that all detours start and end at
  the same position $x$, so detours are done in a non-increasing order of the
  left file.

  We now show that these detours are strictly laminar. Assume there are two
  detours $(a_1,b_1)$ and $(a_2,b_2)$ with $a_1\leq a_2\leq b_1\leq b_2$. After
  the first detour $(a_2,b_2)$ is done, all files between $a_2$ and $b_2$, so
  between $a_2$ and $b_1$ are read. So the second detour $(a_1,b_1)$ can be
  shortened to $(a_1,\leftb {a_2})$ if $a_1<a_2$ or removed if $a_1=a_2$: no
  file is read later and the cost does not increase for any file.

  This concludes the lemma.
\end{proof}

We are now ready to prove the correctness of \DP. This proof relies on an
induction involving several case distinctions ensuring every cost is counted
once. It requires some technical care to precisely define which cost is counted
at each step.

\begin{theorem}
\label{thm:DP}
\DP solves optimally \ltsp in time $O(\nfiles^3 \cdot n)$.
\end{theorem}

\begin{proof}[Proof sketch]
\newcommand{\virtoptabf}[1][b]{\ensuremath{\mathit{VirtOPT_{#1}}(f)}\xspace}
\newcommand{\delay}[2][t_1,t_2]{\ensuremath{\mathit{Delay}_{#1}(#2)}\xspace}

The complexity follows from the dynamic programming definition: there are
$O(\nfiles^2\cdot n)$ cells which are each computed in time $O(\nfiles)$.

We show for all $a,b, \nskip$ by induction on $b-a$ that the computation of
cell $T[a,b,\nskip]$ is correct. Specifically, our induction hypothesis
considers any best solution $S_{a,b,\nskip}$ of the problem given three
additional constraints: (1) there is a detour starting from $a$ and going to
$b$ or a righter file; (2) there is no detour starting between $r(a)$ and
$\ell(b)$ and going to a file righter than $b$; and (3) when the reading head
first reaches $r(b)$, exactly $\nskip$ files have been skipped. Let $t_1$ be
the time when the reading head first reaches $r(b)$ and $t_2$ be the first time
the reading head reaches $r(b)$ (before performing a potential U-turn) after
having read $a$ in $S_{a,b,\nskip}$. For each file $f$, let $t(f)$ the time
when it is served in $S_{a,b,\nskip}$. For each file $f\leq b$, let
$\virtoptabf=r(b)-\ell(f)+s(f)+\uturn$ be the minimum cost to serve $f$ by a
virtual head starting at $r(b)$ and $\virtoptabf=0$ for $b>f$. See
\Cref{fig:proofvirtopt}.

\renewcommand{\addfile}[3][f_\thefile]{
\fill[blue!20!white] (#2,45) rectangle (#3,-315);
\node[blue!50!white] at ({(#3+#2)*0.5}, -355) {$#1$};
}
\begin{figure}[htb]
\resizebox{\linewidth}{!}{
\begin{tikzpicture}[x=\myxscale,y=\myyscale,    my label/.style n args={2}{label={[text=red!50!black, label distance=-2pt]#1:#2}},
pointer/.style = {circle, fill, red!50!black, minimum size = 3pt, inner sep=0}]
\addfile[]{20}{25}
\addfile[]{30}{35}
\addfile[a]{40}{50}
\addfile[f]{60}{65}
\addfile[b]{70}{80}
\addfile[]{90}{95}

\drawLine{70}{80}{40}{95}{20}{35}

\node[pointer, my label={30}{$t_1$}] at (80,-20) {};
\node[pointer, my label={30}{$t_2$}] at (80,-120-3*\uturnpenalty) {};
\node[pointer, my label={-170}{$t(f)$}] at (65,-105-3*\uturnpenalty) {};

\draw[thick, red] (80,-250) -- node [at start, right, yshift=-5, align=left] {\virtoptabf}  (60,-270) -- (60,-270-\uturnpenalty) -- (65,-275-\uturnpenalty);
\draw[ultra thick, green!60!black, dashed] (80,-20) -- (70,-30) -- (70,-30-\uturnpenalty) -- (80,-40-\uturnpenalty) (60,-60-2*\uturnpenalty)--(40,-80-2*\uturnpenalty)--(40,-80-3*\uturnpenalty)--(60,-100-3*\uturnpenalty);
\node[green!60!black] at (55,-20) {\delay{f}};

\draw[ultra thick, |-|] (100,45) -- + (-90,0) node [at start, right] {Tape};

\draw[thick, ->] (8,-50) -- + (0,-80) node [midway, left] {\it time~~};
\end{tikzpicture}}
\caption{Illustration of $t_1$, $t_2$, $t(f)$, \virtoptabf and \delay{f} for $f\leq b$.}
\label{fig:proofvirtopt}
\end{figure}
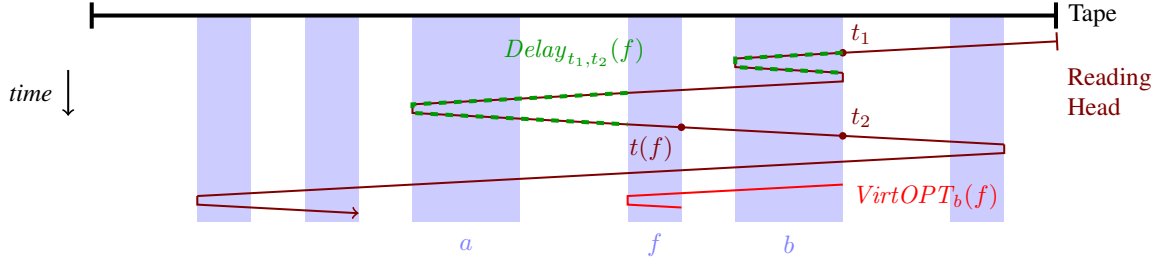

The hypothesis is that cell $T[a,b,\nskip]$ is equal to the sum for all
files $f$ of the impact \delay{f} of what happens between $t_1$ and $t_2$ in
$S_{a,b,\nskip}$ on the service time of $f$, with a basis corresponding to
\virtuallb, {\it i.e.,}

\begin{align}
\label{eq:defdelay}
T[a,b,\nskip] &= \sum_{f} x(f) \cdot  \delay{f}\\
\nonumber\text{with}:\quad
\delay f &:= 0 &\text{ if } t(f)\leq t_1\\
\nonumber\delay f &:= t_2-t_1-\uturn &\text{ if } t(f) > t_2\\
\nonumber\delay f &:= t(f)-t_1-\virtoptabf &\text{ if } t_1 < t(f) \leq  t_2.
\end{align}

Intuitively, for files served after $t_2$, the reading head comes back at the
place it had in $t_1$ at time $t_2$, with the opposite orientation. The delay is
however not equal to $t_2-t_1$ because we should not to count the U-turn
penalty here if a skipped file on the right of $b$ is read within the same
detour starting on $a$. Therefore, the delay equals $t_2-t_1-\uturn$. Counting
the cost based on \virtuallb allowed to simplify the computations in several
places, but in this definition it leads to a less intuitive value of the delay.
For files served between $t_1$ and $t_2$, the file is served at $t(f)$ and we
subtract \virtoptabf to obtain the additional cost on top of the virtual lower
bound.

We now show by induction on $b-a$ that \Cref{eq:defdelay} is correct.
First, consider $T[b,b,\nskip]$ for any $b$, $\nskip$. There are four types of files to consider.
\begin{itemize}
\item $f=b$: we have $t(f)=t_1+2s(b)+\uturn$ and $\virtoptabf=2s(b)+\uturn$ so $\delay{f}=0$,
\item $f>b$ and is not skipped: we have $t(f)\leq t_1$ so \mbox{$\delay{f}=0$},
\item $f>b$ and is skipped: we have $t(f)>t_2$ so $\delay{f}= 2s(b)+\uturn-\uturn$,
\item $f<b$: we have $t(f)>t_2$ so $\delay{f}= 2s(b)$.
\end{itemize}

Overall, there are $\nskip+\nleft b$ files who have a delay equal to $2s(b)$ so:

\begin{align*}
\sum_{f} x(f)\cdot\delay{f} = 2\cdot s(b)\cdot(\nskip+\nleft b) = T[b,b,\nskip].
\end{align*}

This completes the base case of the induction  ($b-a=0$).

\bigskip

Now, consider $T[a,b,\nskip]$ for any values of $a$, $b$ and $\nskip$ such that $a<b$ and assume the induction hypothesis. We want to show that:
\begin{equation}
\label{eq:proofIH}
T[a,b,\nskip] = \sum_{f} x(f)\cdot\delay{f}.
\end{equation}

 We consider
two cases on the structure of $S_{a,b,\nskip}$: either $b$ is served after $a$
or before $a$. 

Assume first $b$ is served after $a$. We want to show that in this case, we have:

\begin{align}
\label{eq:proofskipb}
\sum_{f} x(f)\cdot\delay{f} &=~ \skipb \\
 &=T[a,\leftb{b},\nskip+x(b)] + 2\cdot (r(b)-r(\leftb{b}))\cdot (\nskip+\nleft{a}) \nonumber\\
&+ 2\cdot (\ell(b)-r(\leftb{b}))\cdot x(b) \nonumber.
\end{align}

Let $t_1'$ (resp. $t_2'$) be the first time when the reading head reaches $r(\leftb
b)$ (resp. reaches $r(\leftb b)$ after having read $a$). See \Cref{fig:proofbskip}. So $t_1' = t_1 + r(b)-r(\leftb
b)$ and $t_2' = t_2 - r(b)+r(\leftb b)$. By the induction hypothesis, as (1) there is a detour from $a$ to a file righter than $\leftb b$ (2) there is no detour starting between $r(a)$ and $\ell(\leftb b)$ and going to a file righter than $\leftb b$ (because of the definition of $S_{a,b,\nskip}$ and the assumption that $b$ is served after $a$) and (3) exactly $\nskip + x(b)$ files are skipped at time $t_1'$, we have by the induction hypothesis:

$$T[a,\leftb{b},\nskip+x(b)] = \sum_f  x(f)\cdot\delay[t_1',t_2']{f}.$$

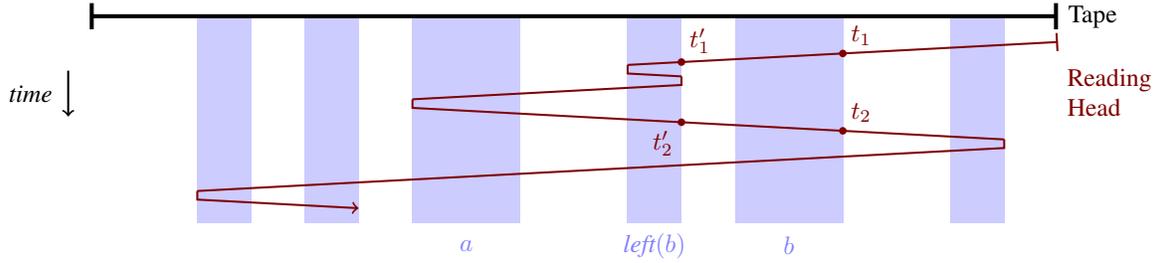
\begin{figure}[htb]
\resizebox{\linewidth}{!}{
\begin{tikzpicture}[x=\myxscale,y=\myyscale,    my label/.style n args={2}{label={[text=red!50!black, label distance=-2pt]#1:#2}},
pointer/.style = {circle, fill, red!50!black, minimum size = 3pt, inner sep=0}]
\addfile[]{20}{25}
\addfile[]{30}{35}
\addfile[a]{40}{50}
\addfile[\leftb b]{60}{65}
\addfile[b]{70}{80}
\addfile[]{90}{95}

\drawLine{60}{65}{40}{95}{20}{35}

\node[pointer, my label={30}{$t_1$}] at (80,-20) {};
\node[pointer, my label={30}{$t_2$}] at (80,-110-3*\uturnpenalty) {};
\node[pointer, my label={30}{$t_1'$}] at (65,-35) {};
\node[pointer, my label={210}{$t_2'$}] at (65,-95-3*\uturnpenalty) {};

\draw[ultra thick, |-|] (100,45) -- + (-90,0) node [at start, right] {Tape};

\draw[thick, ->] (8,-50) -- + (0,-80) node [midway, left] {\it time~~};

\end{tikzpicture}}
\caption{Illustration of $t_1'$ and $t_2'$ when $b$ is skipped. The detour on $\leftb b$ is not required but clarifies the definition of $t_2'$.}
\label{fig:proofbskip}
\end{figure}

We again consider several types of files $f$ to determine \delay{f} in function of \delay[t_1',t_2']{f}.

\begin{itemize}
\item $f=b$: we have $t(f) = t_2'+r(b)-r(\leftb{b})+\uturn$ and $\virtoptabf=2 s(b)+\uturn$ so we have 
\begin{align*}
\delay{f} &= t(f)-t_1-\virtoptabf\\
&= t_2'+r(b)-r(\leftb{b})+\uturn -(t_1'-(r(b)-r(\leftb{b}))) - 2s(b)-\uturn\\
&=t_2'-t_1' +2\cdot (r(b)-r(\leftb b) - s(b)) \\
&= \delay[t_1',t_2']{f} + 2\cdot(\ell(b)-r(\leftb b)).
\end{align*}
\item $f>b$ and is not skipped: $t(f)<t_1$ so $\delay{f} = \delay[t_1',t_2'] f = 0 $.
\item $f>b$ and is skipped: we have $\delay f = t_2-t_1-\uturn = \delay[t_1',t_2']{f} + 2\cdot (r(b)-r(\leftb b))$.
\item $f<a$: same as the previous case.
\item $a\leq f \leq \leftb b$: we have $\virtoptabf = \virtoptabf[\leftb b] + r(b)-r(\leftb b)$ so 
\begin{align*}
\delay f &= t(f)-t_1 - \virtoptabf \\
&= t(f)-t_1 - \virtoptabf + \delay[t_1',t_2']{f} - (t(f)-t_1'-\virtoptabf[\leftb b]) \\
&= \delay[t_1',t_2']{f} +(t_1'-t_1) - (r(b)-r(\leftb b)) \\
&=  \delay[t_1',t_2']{f}.
\end{align*} 
\end{itemize}

Therefore, we obtain \Cref{eq:proofskipb}.

\medskip

Now, assume $b$ is served before $a$. This means that there is a detour from some
file $c>a$ to a file at least as right as $b$. Furthermore, as we assumed that
$S_{a,b,\nskip}$ has no detour from such a file $c$ to a file righter than $b$,
this means that there is a detour $(c,b)$. Therefore, by the laminar property of \Cref{lem:laminar},
and the optimality of $S_{a,b,\nskip}$, there is no detour from a file lefter than $c$ to a file in $[c,b]$.
We want to show that in this case, we have:

\begin{align}
\label{eq:proofrecc}
\sum_{f} x(f)\cdot\delay{f} &=~\detourc\\
\nonumber  &=T[a,\leftb{c},\nskip] + T[c,b,\nskip] \\
\nonumber &\qquad \quad + 2\cdot (r(b)-r(\leftb{c}))\cdot (\nskip+\nleft{a})\\
\nonumber &\qquad \quad +2\cdot\uturn\cdot(\nskip+\nleft{c}).
\end{align}

First we argue that $S_{a,b,\nskip}$ is a solution compatible with the two cells
queried in the expression above. Regarding $T[a,\leftb{c},\nskip]$, we have:
\begin{enumerate}
\item 
a detour from $a$ to a file righter than $\leftb c$, 
\item
no detour from a file
in $[a,\leftb c]$ to a file righter than $\leftb c$ as there is none righter
than $b$ by definition of $S_{a,b,\nskip}$ and there is none between $c$ and $b$
because detours are laminar and there is a detour $(c,b)$, 
\item exactly $\nskip$
files have been skipped when reaching $r(\leftb c)$ as all files between $c$ and
$b$ are read during the detour $(c,b)$. 
\end{enumerate} 
Similarly, regarding $T[c,b,\nskip]$, we
have (1) a detour $(c,b)$ by assumption, (2) no detour from a file in $[c,b]$ to
a file righter than $b$ by definition of $S_{a,b,\nskip}$, and (3) exactly
$\nskip$ files skipped.

We denote by $t_2^0$ the first time $r(b)$ is reached after having read $c$ (before the U-turn penalty),
$t_1'=t_2^0+\uturn+r(b)-r(\leftb c)$ the first time $r(\leftb c)$ is reached and by $t_2' = t_2 -r(b)+r(\leftb c)$ the first time
$r(\leftb c)$ is reached after having read $a$.  Note that $t_1 < t_2^0 <
t_1'<t_2'<t_2$, see \Cref{fig:proofbnotskip}. Therefore, we obtain by the induction hypothesis:

$$T[a,\leftb{c},\nskip] = \sum_f  x(f)\cdot \delay[t_1',t_2']{f} \qquad \text{and} \qquad 
T[c,b,\nskip] = \sum_f  x(f)\cdot \delay[t_1,t_2^0]{f}.$$

\renewcommand{\addfile}[3][f_\thefile]{
\fill[blue!20!white] (#2,45) rectangle (#3,-335);
\node[blue!50!white] at ({(#3+#2)*0.5}, -375) {$#1$};
}
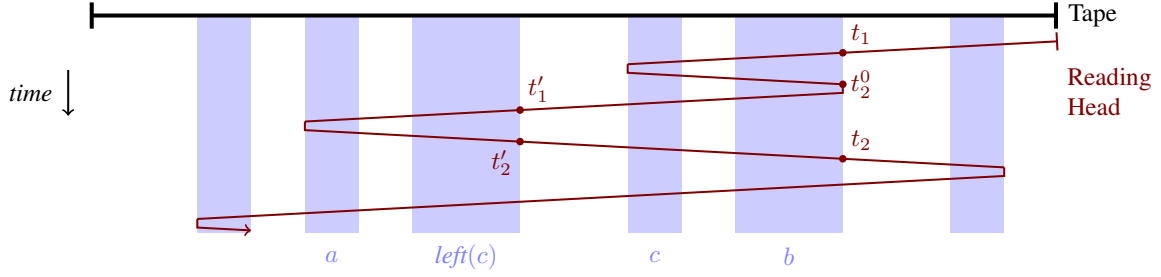
\begin{figure}[tb]
\resizebox{\linewidth}{!}{
\begin{tikzpicture}[x=0.15cm,y=0.008cm,    my label/.style n args={2}{label={[text=red!50!black, label distance=-2pt]#1:#2}},
pointer/.style = {circle, fill, red!50!black, minimum size = 3pt, inner sep=0}]
\addfile[]{20}{25}
\addfile[a]{30}{35}
\addfile[\leftb c]{40}{50}
\addfile[c]{60}{65}
\addfile[b]{70}{80}
\addfile[]{90}{95}

\drawLine{60}{80}{30}{95}{20}{25}

\node[pointer, my label={30}{$t_1$}] at (80,-20) {};
\node[pointer, my label={0}{$t_2^0$}] at (80,-60-\uturnpenalty) {};
\node[pointer, my label={30}{$t_2$}] at (80,-160-3*\uturnpenalty) {};
\node[pointer, my label={30}{$t_1'$}] at (50,-90-2*\uturnpenalty) {};
\node[pointer, my label={210}{$t_2'$}] at (50,-130-3*\uturnpenalty) {};

\draw[ultra thick, |-|] (100,45) -- + (-90,0) node [at start, right] {Tape};

\draw[thick, ->] (8,-50) -- + (0,-80) node [midway, left] {\it time~~};

\end{tikzpicture}}
\caption{Illustration of $t_2^0$, $t_1'$ and $t_2'$ when $b$ is not skipped.}
\label{fig:proofbnotskip}
\end{figure}

We again consider several types of files $f$ to determine $\delay{f}$:

\begin{itemize}
\item $f>b$ and is not skipped: all delays equal zero as $t(f)<t_1$.
\item $c\leq f\leq b$: we have  $\delay{f} = t(f) - t_1 - \virtoptabf = \delay[t_1,t_2^0]{f}$ as $t(f) \leq t_2^0$ and $\delay[t_1',t_2']{f}=0$ as $t(f)<t_1'$.
\item $a\leq f\leq \leftb c$: we have: 
\begin{align*}
\delay f &= t(f)-t_1 -\virtoptabf \\
&= t(f)-t_1'+t_1'-t_1 - \virtoptabf[\leftb c] - (r(b)-r(\leftb c)) +t_2^0-t_2^0-2\uturn+2\uturn\\
&=  \left(t(f)-t_1'-\virtoptabf[\leftb c]\right) + \left(t_1'-t_2^0-\uturn\right) - \left(r(b) -r(\leftb c)\right) \\
&\qquad\quad+ \left(t_2^0-t_1-\uturn\right) + 2\uturn\\
&= \delay[t_1',t_2']{f} + 0+ \delay[t_1,t_2^0]{f}+2\uturn.
\end{align*}
\item $f<a$: we have:
\begin{align*}
\delay{f} &= t_2-t_1-\uturn +2\uturn-2\uturn +t_2'-t_2'+t_1'-t_1'+t_2^0-t_2^0\\
&=  (t_2'-t_1'-\uturn) + (t_2-t_2') + (t_1'-t_2^0-\uturn) +(t_2^0-t_1-\uturn) + 2\uturn\\
&= \delay[t_1',t_2']{f} + 2\cdot (r(b)-r(\leftb c)) + \delay[t_1,t_2^0]{f}+ 2\uturn.
\end{align*}

\item $f>b$ and is skipped: same as the previous case.
\end{itemize}

Therefore, we get \Cref{eq:proofrecc}.

\medskip

We now conclude the proof of the induction.

As $S_{a,b,\nskip}$ must either serve $b$ before $a$ or include a detour $(c,b)$ as argued earlier, we have:

$$ \mathit{cost}(S_{a,b,\nskip}) \geq \min\big( ~ \textit{skip}(a,b,\nskip) ~; ~
\min_{c\in \fab} ~  \mathit{detour_c}(a,b,\nskip) \big) = T[a,b,\nskip].$$

And we get the equality by optimality of $S_{a,b,\nskip}$.

\bigskip

Finally, we get by induction, for all $a$, $b$, $\nskip$ and $S_{a,b,\nskip}$:

$$T[a,b,\nskip] := \sum_{f} x(f) \cdot  \delay{f}.$$

Note that $S_{f_1,f_{n_f},0}$ is equal to the optimal solution of the problem.
So, denoting by $t_0$ the starting time of the solution and $t_{\max}$ the time
at which the reading head would reach back the right of the tape in
$S_{f_1,f_{n_f},0}$ (it may stop earlier if the rightmost file is not skipped),
we get that the content of the cell $T[f_1,f_{n_f},0]$ is equal to:

\begin{align*}
T[f_1,f_{n_f},0] &= \sum_{f} x(f) \cdot  \delay[t_0,t_{\max}]{f}\\
&= \sum_f x(f)\cdot (t(f)-t_0 - \virtoptabf[f_{n_f}])\\
&= \mathit{cost}(S_{f_1,f_{n_f},0}) - \virtuallb.
\end{align*}

Therefore, we obtain that the optimal cost is equal to $\mathit{OPT}=T[f_1,f_{n_f},0]+\virtuallb$, which completes the proof.
\end{proof}

\subsection{Efficient heuristics}

\newcommand{\logparam}{\ensuremath{\lambda}\xspace}

The complexity of \DP may be prohibitive for an input containing hundreds of
requested files. We address this issue by providing two lighter algorithms named
\logdp and \simpledp. Both restrict the dynamic program search space,
in two different ways, in order to propose a suboptimal solution in a shorter time.

\subsubsection{Restricting the detours length: \logdp}

\logdp is equal to \DP except that when computing \detourc, $c$ is
restricted to be at most $\logparam\cdot\log \nfiles$ requested files apart from~$b$, for a
constant parameter $\logparam$. This reduces both the table dimensions and
complexity to query a single cell ans thus leads to a time complexity of $O(\nfiles\cdot n\cdot \log^2 \nfiles)$. Only detours of span at most $\logparam\cdot
\log \nfiles$ are then considered, and the solution returned is optimal among
this class of schedules. The parameter $\lambda$ can be adjusted to trade
accuracy for computing time. As this solution is by definition at
least as good as \gs, it is also a 3-approximation if $\uturn=0$.

We remark that the approximation ratio of \logdp is actually equal to 3 if
$U=0$, no better than the one of \gs. Indeed, consider an arbitrarily large
integer $z$ and an instance with $z$ requested files. The leftmost file $f_1$ is
small and non-urgent, $\ell(f_1)=0$, $s(f_1)=1$ and $x(f_1)=1$. The $z-1$ other
files are located far on the right and are contiguous, $\ell(f_{2+i}) = 2z^3+i$
for all $i<z-1$. All these files have a unit size except the rightmost one which
is large: $s(f_{2+i})=1$ for all $i<z-2$ and $s(f_z)=z^2$. Finally, $f_2$ is
urgent, $x(f_2)=z^2$, $f_z$ is less urgent, $x(f_z)=z$ and all other files have
exactly one request. The optimal solution has a single detour $(f_2,f_z)$ before
reading $f_1$ and has then a cost equal to $C_\opt = z^4+O(z^3)$, the $z^4$
coefficient coming from the requests associated to $f_2$. If detours spanning
$z-1$ files are forbidden, then we study two complementary cases. If $f_z$ is
read before $f_2$, then $f_2$ is read after a time larger than $3s(f_z)$ which
incurs a cost of $3z^4 = 3\cdot C_\opt - o(C_\opt)$. Otherwise, $f_z$ is read
after $f_1$, so after a time at least $2\ell(f_2)$ which incurs a cost of
$4z^4$. Hence, \logdp cannot have an approximation ratio smaller than 3. With an
arbitrary value of $U$, the approximation ratio is infinite as the restriction
on the detours length can lead to having to resort to many detours, consider the
example above with equivalent files for $f_2,\dots, f_z$.

\subsubsection{Forbidding intertwined detours: \simpledp}

\simpledp simplifies \DP in another aspect to reduce its complexity. It
restricts the search space to solutions in which all detour intervals are
disjoint: no file is traversed from the left to the right after having being
read, except possibly at the last phase after the leftmost file has been read.
The implementation of this modification is done by simply modifying the \detourc
function. Instead of using a recursive call to compute the optimal strategy
between $c$ and $b$ if there is a detour $(c,b)$, it is now possible to directly
incur the cost of the detour $(c,b)$ as no subsequent detour is allowed inside
this interval. This cost corresponds to the length of the detour for requests on
the left of $c$ and to the distance between $c$ and $f$ for any file $f$
requested between $c$ and $b$:
\begin{align*}
&\detourc := T[a,\leftb c,\nskip]\\
&~~~~+ 2\cdot (r(b)-r(\leftb c))\cdot (\nskip+\nleft{a})\\
&~~~~+ 2\cdot (U+r(b)-\ell(c))\cdot (\nskip+\nleft{c})\\
&~~~~+ \sum_{c<f\leq b} 2\cdot(\ell(f)-\ell(c))\cdot x(f).
\end{align*}

Consequently, the first index ($a$) of the dynamic program table becomes useless as it is
always equal to $f_1$, the leftmost requested file. The complexity of this
algorithm is then in $O(n\cdot\nfiles^2)$.

Contrarily to \logdp, we conjecture that the approximation ratio of \simpledp is
better than the factor~$3$ inherited from the greedy algorithm \gs when $U=0$.
Specifically, we exhibit an example showing that the approximation ratio is at
least $5/3$ and show that for any value of $U$, it is at most $3$. We believe
that the approximation ratio actually equals $5/3$.

\begin{lemma}
\label{lem:simpledp} The approximation ratio of \simpledp belongs to $[5/3,3]$ for any value of $U$.
\end{lemma}

\begin{proof}

We first provide an instance on which the solution of \simpledp approaches
$5\opt/3$. We then prove that it never exceeds $3\opt$ for any value of $U$.

Consider an instance parameterized by a large integer $z$ with four requested
files $f_1$, $f_2$, $f_3$, and $f_4$. Let $\ell(f_1)=0$, $s(f_1)=1$ and
$x(f_1)=1$, this file is used to ``force'' the rightmost files to be read using
detours before reaching $f_1$. The three other files are located far on the right, $\ell(f_2)=3z^2$.
The files $f_2$ and $f_3$ are urgent, small, and separated: $s(f_2)=s(f_3)=1$,
$x(f_2)=x(f_3)=z^2$ and $\ell(f_3)=r(f_2)+z$. Finally, the file $f_4$ is large,
less urgent, and contiguous to $f_3$: $\ell(f_4)=r(f_3)$, $s(f_4)=z$, and
$x(f_4)=z$. The right end of the tape corresponds to the right of $f_4$. One
solution involving intertwined detours is to read first the small file $f_3$,
then $f_2$ and $f_4$ in the same detour before reading $f_1$, see \Cref{fig:LB}
for an illustration. The cost of this solution equals:

$$C_\opt:= x(f_2)\cdot (r(f_4)-\ell(f_2)) + x(f_3)\cdot s(f_4) + O(z^2) = 3z^3 + O(z^2).$$

\renewcommand{\addfile}[3][f_\thefile]{
\addtocounter{file}{1}
\fill[blue!20!white] (#2,45) rectangle (#3,-335-7*\uturnpenalty);
\node[blue!50!white, align=center] at ({(#3+#2)*0.5}, -380-7*\uturnpenalty) {$#1$};
}
\begin{figure}[h!]
\centering
\resizebox{\linewidth}{!}{
\begin{tikzpicture}[x=.5*\myxscale,y=.7*\myyscale]
\setcounter{file}{0}
\addfile[x(f_1)=1]{-100}{-97}
\addfile[x(f_2)=z^2~~~~~]{28}{31}
\addfile[x(f_3)=z^2~~~~~~~~]{61}{64}
\addfile[~~~~~~~~~~~x(f_4)=z]{65}{95}

\drawLine{61}{64}{28}{95}{-100}{-97}

\draw[ultra thick, |-|] (100,45) -- + (-210,0) node [at start, right] {Tape};

\draw[thick, ->] (-110,-50) -- + (0,-80) node [midway, left] {\it time~~};

\draw[black!80] (-97,70) edge[<->] node[midway,above] {$3z^2$}  (28,70)
                (31,70) edge[<->] node[midway,above] {$z$}  (61,70)
                (65,70) edge[<->] node[midway,above] {$z$}  (95,70);

\foreach \i in {1,2,...,10}
{
	\fill[white] (-77+2*\i,72) rectangle ++(1,-600-7*\uturnpenalty);
}

\end{tikzpicture}
}
\caption{Instance exhibiting a lower bound on the approximation ratio of \simpledp.}
\label{fig:LB}
\end{figure}
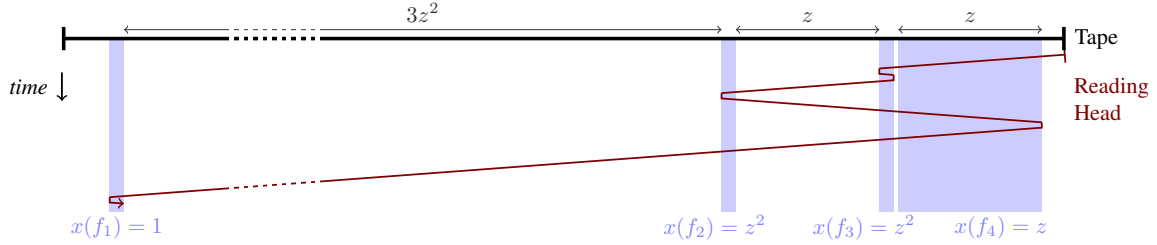

We then show that all solutions without intertwined detours have a cost of at
least $\frac{5}{3} C_\opt + O(z^2)=5z^3+O(z^2)$. We do a case analysis based on
which detour $f_4$ is read on.

\begin{itemize}
\item $f_4$ is read in the detour $(f_4,f_4)$: $f_3$ is read after $3s(f_4)=3z$
and $f_2$ after $r(f_4)-\ell(f_2)+2s(f_4)>4z$ so the cost exceeds $7z^3$.

\item $f_4$ is read in the detour $(f_3,f_4)$: $f_3$ is read after $s(f_4)=z$
and $f_2$ after $r(f_4)-\ell(f_2)+2s(f_4)>4z$ so the cost exceeds $5z^3$.

\item $f_4$ is read in the detour $(f_2,f_4)$: $f_3$ must be read in that same
detour as intertwined detours are forbidden. So $f_3$ is read after
$r(f_4)-\ell(f_2)+r(f_3)-\ell(f_2)>3z$ and $f_2$ after $r(f_4)-\ell(f_2)>2z$ so
the cost exceeds $5z^3$.

\item $f_4$ is read in the detour $(f_1,f_4)$: the cost associated to the
requests on $f_4$ exceeds $x(f_4)\cdot2\cdot\ell(f_2)=6z^3$. 
\end{itemize}

As $z$ grows, this shows that the approximation ratio of \simpledp is at least $5/3$.

\medskip

We now prove the second part of the lemma: for all values of $U$, the
approximation ratio of \simpledp is at most $3$. As noted above, this result is
already known for $U=0$, as the solution is at least as good as the one taking
all atomic detours.

Consider any instance of \ltsp and an optimal solution of cost $\opt$ described
by a list of strictly laminar intertwined detours $L$, such as the one returned
by \DP. We iteratively modify the solution $L$, reducing the portion of tape
witnessing intertwined detours while guaranteeing that the final cost does not
exceed $3\opt$. We again assume that the final detour $(f_1,f_{n_f})$ is not
explicitly present in $L$.

We say that a detour $(a,b)\in L$ is \emph{major} if there exists a detour
$(f_i,f_j)\in L$ such that $a<f_i\leq f_j<b$. Any such detour $(f_i,f_j)$ is
said to be \emph{inside} $(a,b)$. Among the major detours of $L$, consider the
one with the rightmost right endpoint. Let this detour be $(a,b)$. Then, among
the detours inside $(a,b)$, consider the one with the rightmost left endpoint.
Let this detour be $(c,d)$. We then have $a<c\leq d < b$. 

We can then split the schedule induced by $L$ into three time periods. First,
the files on the right of $b$ are read using non-major detours or skipped until
the final detour $(f_1,f_{n_f})$. Then, the files located between $a$ and $b$
are all read: the first one to be read is $c$ by definition and the last one is
$b$. Then, the files on the left of $a$ are read, and finally the remaining ones
on the right of $b$ are read.

We modify $L$ as follows: the detour $(a,b)$ is replaced by $(a,\leftb{c})$ and
the detour $(c,d)$ is replaced by $(c,b)$, where $\leftb{c}$ represents
the closest requested file located at the left of $c$. The consequences are the
following:
\begin{itemize}
\item files in $[c,d]$ are read at the same time as the original solution. 

\item files in $[d,b]\setminus\{d\}$ are read sooner as part of the detour $(c,b)$.

\item files read after $a$ in the original solution are read sooner as the
number of detours did not change but the distance traversed decreased.

\item for files in $[a,\leftb{c}]$, the reading head now performs the detour
$(c,b)$ instead of $(c,d)$ before reading them. This incurs an additional time
of $2(r(b)-r(d))$.

\item there is no major detour going over the file $c$ or a file on its right.
\end{itemize}

A simple upper bound is that the cost increases by at most $\nleft{c}\cdot
2\cdot(r(b)-r(d))$, where $\nleft{c}$ represents the number of file requests
located on the left of $c$, excluding $c$.

Consider successive applications of this process until no major detour is left.
This is always possible as, after each step, the rightmost right endpoint of a
major detour is moved to the left. This leads to the following sequence of files
involved in the modified detours: $\{(a_i,c_i,d_i,b_i)\}_{i\in [1,n_d]}$. After
each application at step $i$, the new rightmost right endpoint of a major
detour, $b_{i+1}$, is located on the left of $c_i$, so of $d_i$. This means that
the intervals $\{[d_i,b_i]\}_{i\in[1,n_d]}$ are all pairwise disjoint. Therefore, the additional cost is at most:

$$\sum_{i=1}^{n_d} \nleft{c_i}\cdot 2\cdot(r(b_i)-r(d_i)) \leq 2\cdot  \sum_{j=1}^{n_f} x(f_j) \cdot (m-r(f_j)) \leq 2\cdot \opt.$$

The first inequality comes from the fact that, for each file request, the union
of the relevant intervals $[b_i,d_i]$ represents a subset of the part of the
tape located on the right of this file.

Therefore, the final cost of the solution obtained, free of intertwined detours,
is at most $3\cdot\opt$, which proves the lemma.
\end{proof}

\section{Performance evaluation}
\label{sec:expe} 

In this section, we evaluate the performance, as the sum of service times of
its generated sequence of detours, of our exact algorithm, \DP, and its
suboptimal versions \simpledp and \logdp with a reduced complexity on a real-world dataset. We also compare  the
performance of these algorithms to existing ones~\cite{cardonha2018} (see
\Cref{sec:prelim}). 
Aiming for reproducibility, the source code used in this section\footnote{\url{https://figshare.com/s/80cee4b7497d004dbc70}} and the dataset\footnote{\url{https://figshare.com/s/a77d6b2687ab69416557}} are freely available online.

\subsection{Evaluated algorithms}
\label{sec:heur}
We consider \simpledp and two variants of \logdp with different values of the \logparam
parameter, 1 and 5, that we denoted by \logdpone and \logdpfive. Then, we
adapted the \fgs, \nfgs, and \lognfgs algorithms from~\cite{cardonha2018} to
take U-turn penalties into account. We further modified \nfgs on three points
which we believe were intended by the original authors as otherwise \nfgs may
not be as good as \fgs, a property which was claimed in the paper. Details
concerning our implementation can be found in Appendix~\ref{app:algos} and in
the source code. All these algorithms were implemented in a single-thread Python
program.

For each tape, each algorithm needs the following inputs:
\begin{itemize}
\item an ordered list of indices of the files requested on the tape
\item the number of requests for each requested file
\item the size of all files on the tape
\item the cost of the U-turn penalty
\end{itemize}
The output of an algorithm is a list of detours where a detour is a couple
$(a,b)$ which means that the head goes to the left of file $a$ then to the right
of file $b\ge a$. A value of $a=0$ corresponds to the leftmost requested file on
the tape. Then, we compute the sum of service times for each file request
following the sequence of detours given by each algorithm.

\subsection{Inputs from production logs}
\label{sec:input}

The \CC, from which our dataset comes, uses tape storage for long-term projects
in High Energy Physics and Astroparticles physics. Its tape library is
currently composed of 48 TS1160 drives and can store up to 6,700 20TB IBM
Jaguar E tapes.

The raw dataset covers two weeks of activity. It contains millions of lines of
reading, writing, and update requests with their associated timestamp. 
We applied several filtering steps to obtain the inputs needed by the
algorithms. We restricted to reading requests, and
selected a set of 169 tapes of interest storing $3,387,669$ files. Each tape is
divided into segments whose size and number depend on the tape.  In a segment,
files and \emph{aggregates} of files are described by several features such as
a position and a size. An aggregate is a batch of related files that can be
written sequentially.  A segment contains an aggregate if there is more than
one file referenced in this segment. Within an aggregate, the position of a
file is described a couple (position,
offset) %
where the position corresponds to the beginning of the aggregate, thus the
beginning of a segment, and the offset is the relative position of the file
within the aggregate. Note that an aggregate can span across several
segments. We discarded such aggregates and their associated requests to focus
on aggregates lying on a single segment. Reading files inside an aggregate is
not straightforward and generates a non-negligible overhead as the head
is required to go to the start of the aggregate before reading a file.

Finally, we decided to consider that requesting a file within an aggregate will
be treated as a request to read the whole aggregate. While this simplifies log
filtering process, this assumption also corresponds to a common optimization
strategy. Read aggregates are stored on disks when a file it contains is read
for the first time. Then, all the subsequent accesses to files in this
aggregate will avoid the large delays induced by tapes and benefit of the
smaller latency of disks. Consequently, we replace all the file requests in
a given aggregate by a single request for a file of the size of
this aggregate. Then we associate to this file a number of requests equal to
the number of requested files in that aggregate.

To summarize, the processed dataset corresponds to a total of $119,877$ files
stored on the 169 tapes. We provide more details and statistics on this dataset
in Appendix~\ref{sec:repro}. To the best of our knowledge, this is the first time
that a  realistic dataset for magnetic tape storage is made publicly
avaible. In the context of the evaluation of the considered algorithms, this
dataset corresponds to 169 distinct instances of \ltsp to solve.

\subsection{Simulation results}
\label{sec:results}

The evaluations presented in this section have been performed on a single
server with two Intel Xeon Gold 6130 CPUs with 16 cores each.  To compare the
performance of the different algorithms, we use the generic \emph{performance
  profile} tool~\cite{perfprofile}. We compute the cost of each algorithm on
each instance of the dataset, normalize it by the optimal (\DP), and report an
empirical cumulative distribution function. For a given algorithm and an
overhead $\tau$ expressed in percentage, we compute the fraction of instances
for which the algorithm has a cost at most $(1+\tau)\cdot\mathit{cost}(\DP)$,
and plot these results. Therefore, the higher the curve, the better the method.
For instance, for an overhead of $\tau=10\%$, the performance profile shows how
often the performance of a given algorithm lies within $10\%$ of the optimal
solution.

We evaluate the algorithms on each of the 169 instances for three different values of the U-turn penalty \uturn: 
(i) no penalty
(ii) a penalty equals to half of the average size of a segment in the 169 considered tapes, and 
(iii) a penalty equivalent to the average size of a segment.
While we have not yet modeled seeking and reading speeds of the head, such penalties whose values are extracted from 
features of the input instances are useful to evaluate the impact of increasing
\uturn on the  performance of the algorithms.

\paragraph{Algorithms Performance}

Figure~\ref{fig:U0} shows the performance profiles of the algorithms without
U-turn penalty.  As expected, \gs and \nodetour show poor performance, with an
overhead of more than 10\% for \nodetour over 60\% of the instances. The \fgs,
\nfgs, and \lognfgs heuristics exhibit very similar performance, with an
overhead of less than 2.5\% over 80\% of the test cases. Both variants of \logdp
heuristic slightly outperform the other heuristics, and \simpledp is the best
solution by a greater margin. As expected, the higher \logparam, the closer to
optimal the solution is. \nfgs is better than $\logdp(1)$ on $11\%$ on the
instances, and worse in $85\%$. It performs better when a single long detour is
largely beneficial, and out of reach of $\logdp$. \nfgs is slightly better than
\simpledp on $<4\%$ of the instances, where a large intertwined detour is more
beneficial.

\begin{figure}[hbtp]
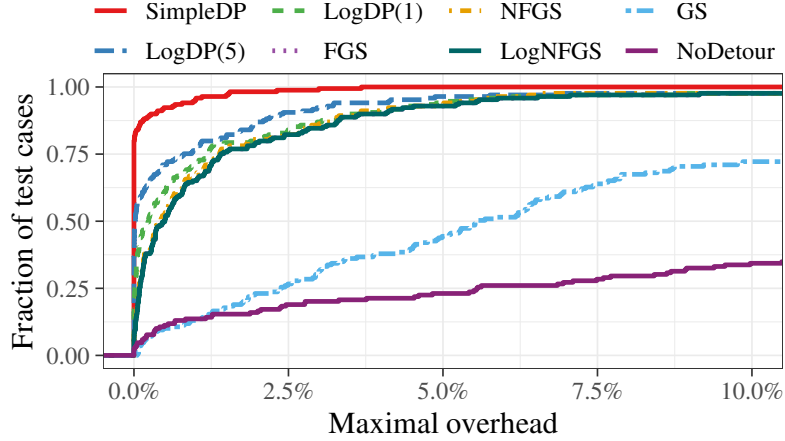

\centering
\resizebox{.7\linewidth}{!}{

 }
\caption{Performance of the different algorithms, when $\uturn = 0$.}
\label{fig:U0}
\end{figure}

Figure~\ref{fig:U1} illustrates the algorithms performance with a U-turn penalty
equal to the average size of a segment.  We see that $U$ increases the
discrepancy between the \fgs-like heuristics  and \logdp and \simpledp. 
Here, these heuristics cause at least 5\% more overhead on half of the instances
than \logdpone, and 10\% more overhead than \simpledp.  The
suboptimal solutions of \DP variants are more robust to the increase of
\uturn, with an overhead of less than $1\%$ for \simpledp when compared to \DP
for 97\% of the inputs. Similar trends can be observed with a halved value of $U$ on \Cref{fig:U.5}.

\begin{figure}[htb]
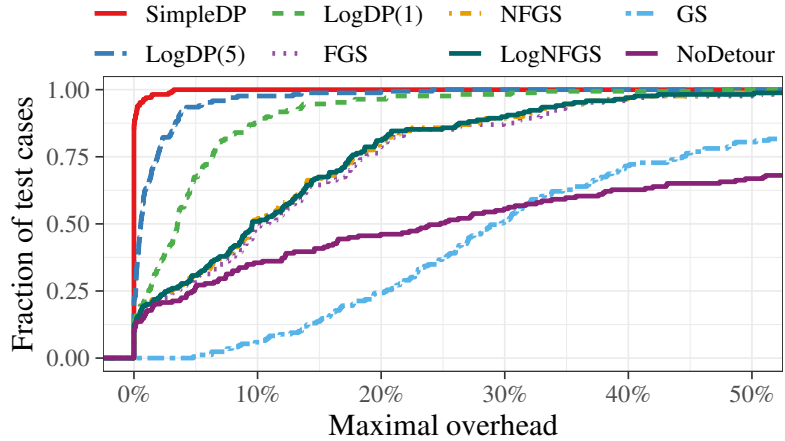

\centering
\resizebox{.7\linewidth}{!}{

 }
\caption{Performance of the different algorithms, when $\uturn$ is equal to the average segment size.}
\label{fig:U1}
\end{figure}

\begin{figure}[htb]
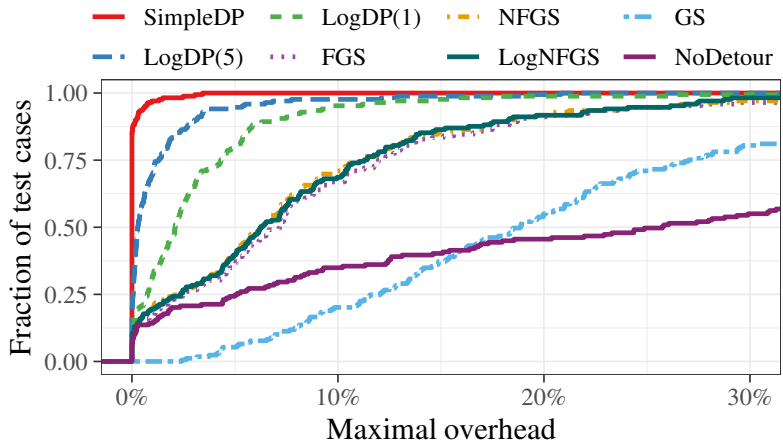

\centering
\resizebox{.7\linewidth}{!}{
%
 }
\caption{Performance of the different algorithms, when $\uturn$ is equal to half the average segment size.}
\label{fig:U.5}
\end{figure}

\paragraph{Time to solution} The median running times for the algorithms \DP,
{\logdp}(5), \simpledp, {\logdp}(1), \nfgs and \lognfgs are around 281, 47, 21, 5, $0.4$ and
$0.1$ seconds respectively. The other algorithms have insignificant running
times (\mbox{$<\!1$ms}). However, our single-thread Python implementation was not designed
with performance in mind. Estimations based solely on the documented maximum
speed of the reading head leads to an average duration of $500$s to schedule the
requests on one tape of the dataset with an average service time of $80$s. The
observed gains thus have to be nuanced by the required computing times of the
algorithms. It should also be noticed that the schedule
computation can be done in parallel to robot operations mounting the tape, so
the start of the schedule is not directly delayed by the computation time. The
characteristics of the data set (a median $n>2,600$ much larger than
$\nfiles<150$) also explain the longer running times of \DP variants as the \fgs-like
algorithms complexity does not depend on $n$, see more details in the
supplementary material. The $\lambda$ parameter can be used to obtain a faster
version of $\logdp$ at the cost of lower performance. On large inputs (\ie list
of requested files greater than 100), the cost of \DP becomes prohibitive in a
production context, making \logdp variants good replacement candidates.

\section{Conclusion}

In this article we studied the Linear Tape Scheduling Problem, aiming at
minimizing the average service time for read requests on a linear magnetic
tape. We proposed an exact polynomial-time dynamic programming algorithm,
solving this problem whose complexity was open until now. Then, we derived a
low-cost suboptimal algorithm, whose performance outperforms existing heuristics on a
realistic dataset extracted from the tape library logs of the \CC, a dataset we make publicly available.

This dataset could also be used for related problems such as $k$-server on the
line for which few relevant datasets are available~\cite{lindermayr2021double}.
The remaining question on the theoretical side of \ltsp resides in the possible
improvements in the running time of an exact algorithm. Notably, as discussed in
\Cref{sec:framework}, the input of \ltsp is defined as a list of requests,
possibly on duplicate files. If the number of requests is not bounded by a
polynomial in the number of requested files, this is not the best representation
of the input. It would be more compact to define the input as a set of requested
files associated with the number of requests on each file. The algorithms \DP,
\logdp and \simpledp would then be only pseudo-polynomial in this setting as
they are not polynomial in $\log n$. Therefore, the complexity of this problem
is still open. Another interesting question resides in the determination of the
approximation ratio of \simpledp, which belongs in $[5/3,3]$ for any value of
$U$. In other words, the question is to determine the exact gain of using
intertwined detours. The obvious generalization of the problem would be to
consider the two-dimensional tape geometry, but we expect that such a model
would quickly become intractable. We also discuss below how \DP can be adapted
to handle two minor extensions: arbitrary starting position of the head and a
different reading speed.

\paragraph*{Arbitrary starting position.}
The starting position of the reading head could be chosen at an arbitrary
position $X$ and the algorithm \DP can be adapted to find the optimal solution:
simply prevent any detour to start on the right of $X$. Indeed, this emulates a
schedule in which the head initially moves from the rightmost file to $X$. No
detour starting on the right of $X$ would ever be needed later thanks to
\Cref{lem:laminar}.

\paragraph*{Different reading speed.}
We do not differentiate seeking speed, where the tape is required to move to a
specific location, and reading speed, where data is actually output. The model
could be tuned to accept such two different speeds, but we chose to keep it
simpler by using a unique speed. This choice is motivated by the observation
that reading times are much smaller than seeking times in the tapes operated in
the studied computing center. \DP could be easily transformed to account for such
different speeds. The only limitation being that \DP would require to
read each file the first time it is traversed from left to right, which means
that the solution returned would not be optimal on adversarial inputs requiring
multiple back-and-forth seeks over a file before reading it.

\section*{Acknowledgments}
We thank Pierre-Emmanuel Brinette for fruitful discussions.
Experiments presented in this paper were carried out using the Grid'5000
testbed, supported by a scientific interest group hosted by Inria and including
CNRS, RENATER and several Universities as well as other organizations (see
\url{https://www.grid5000.fr}).

\bibliographystyle{plainnat}
\bibliography{biblio}

\renewcommand{\addfile}[3][f_\thefile]{
\addtocounter{file}{1}
\fill[blue!20!white] (#2,45) rectangle (#3,-315);
\node[blue!50!white] at ({(#3+#2)*0.5}, -355) {$#1$};
}

\appendix

\section{Relationship with the concurrent work~\cite{cardonha21}}

\label{app:cardonha21}

Concurrently to this study, Cardonha, Ciré and Real~\cite{cardonha21} achieved similar
results on the same Linear Tape Scheduling Problem. They also provide a
polynomial-time quartic algorithm based on dynamic programming which resolves
the complexity status of the problem. However, our results differ in several points:
\begin{itemize}
\item their model considers a single request per file,
\item we introduced the U-turn penalty $U$ to account for mechanical deceleration,
\item their dynamic programming formulation relies on two inter-connected tables whereas our algorithm uses a single table,
\item they propose different heuristics: an approximate variant of the dynamic programming with lower constant factors and a greedy heuristic to exchange some files read order based on their size,
\item they compare heuristic performances on synthetic data to determine if some parameters used in the instances generation influences the results,
\item the realistic dataset they use has a very low variance per file and one request per file. This means that any heuristic based on Greedy Scheduling is optimal~\cite{cardonha2018}. The dataset we use presents a broad spectrum of file size variance and number of requests per file.
\end{itemize}

\newcommand{\mc}[1]{\ensuremath{\mathcal #1}\xspace}
\newcommand{\files}{\ensuremath{\mc F}\xspace}
\newcommand{\requests}{\ensuremath{\mc R}\xspace}
\newcommand{\tape}{\ensuremath{\mc T}\xspace}
\newcommand{\detours}{\ensuremath{\mc L}\xspace}
\newcommand{\resu}{\ensuremath{\mathit{res}}\xspace}
\newcommand{\tempo}{\ensuremath{\mathit{temp}}\xspace}
\newcommand{\wasdetour}{\ensuremath{\mathit{WasADetour}}\xspace}
\newcommand{\rightestdetour}{\ensuremath{\mathit{RightestDetour}}\xspace}
\newcommand{\cost}{\ensuremath{\mathit{cost}}\xspace}

\section{Precise description of the algorithms adapted from \cite{cardonha2018}}
\label{app:algos}

Each algorithm considered in this section takes the following inputs:
\begin{itemize}
\item an ordered list \files of indices of the files requested on the tape,
\item the number of requests \requests for each requested file,
\item the size of all files on the tape \tape,
\item the cost of the U-turn penalty \uturn.
\end{itemize}
The output of an algorithm is a list of detours where a detour is a couple
$(a,b)$ which means that the reading head goes to the left of file $a$ then to
the right of file $b\ge a$. A value of $a=0$ corresponds to the leftmost requested file on
the tape.

We adapted \fgs, \nfgs and \lognfgs from~\cite{cardonha2018} to take into
account U-turn penalties. We also modified \nfgs on three points which we
believe were intended by the original authors as otherwise \nfgs may not be as
good as \fgs, a property which was claimed in the paper.

The pseudo-code depicted in this section is rather high-level, referring to
mathematical inequalities without expliciting how to maintain each term. We
explain the time complexity of our implementation and the low-level details can
be checked directly in the source code.

\subsection{Restating structural results}

Before describing the algorithms, we need some preliminary definitions and
results, on which the algorithms rely.

We say that a file $f$ belongs to list of detours $\detours$ if and only if it
is part of a detour of \detours:
$$ f\in\detours ~\Leftrightarrow~ \exists (a,b)\in \detours~|~ a\leq f\leq b.$$
We assume that the tape starts at a requested file on its left to simplify the
formulas: the reading head will have to go to the position 0, so at a distance
$\ell(f)$ from the left of any file $f$ (this assumption allows to drop additive $-\ell(f_1)$ terms).

The first result will be used by the algorithm \fgs.

\begin{lemma}
Let \detours be a list of single-file detours $(f_i,f_i)$ and $f$ be a file such
that $(f,f)\in\detours$. Then, $\cost(\detours\setminus\{(f,f)\})<\cost(\detours)$ if and only if:
\begin{equation}
\label{eq:fgs}
2\cdot x(f)\cdot\left(\ell(f) + \sum_{ g<f~|~g\in\detours}\left(s(g)+\uturn\right)\right)
<
2\cdot (s(f)+\uturn)\cdot\left(\sum_{g<f}x(g) + \sum_{g>f ~|~ g\notin\detours} x(g)\right).
\end{equation}
\end{lemma}

\begin{proof}
This equation with $\uturn=0$ corresponds to Corollary 4 in \cite{cardonha2018}.

The left-hand side equals the delay added to the service time of $f$: for each
request of $f$, the reading head has to go the left of the tape ($2\ell(f)$) and
through all the detours $(g,g)\in\detours$ on the left of $f$, where each
detour adds a delay of $2(s(g)+\uturn)$.

The right-hand side corresponds to the delay added to all other files than $f$
by performing a detour of duration $2(s(f)+\uturn)$ to serve $f$. The impacted
files are the ones at the left of $f$ and the skipped files.
\end{proof}

We now define the function $\Delta$ required by the algorithm \nfgs.

\begin{definition}
Let \detours be a list of detours and $(a,b)$ be a detour such that no detour in
\detours starts on $a$. We define:
\begin{align*}
\Delta(\detours, (a,b)) ~=~ 
&2\cdot(r(b)-\ell(a)+\uturn)\cdot\left(\sum_{f<a}x(f) + \sum_{f>b~|~f\notin\detours}x(f)\right)\\
&~~-
2\sum_{f\in [a,b] ~|~ f\notin\detours} x(f)\cdot\left(\ell(a)+\sum_{(f',\,g')\in\detours~|~f'<a}(r(g')-\ell(f')+\uturn)\right).
\end{align*}
\end{definition}
\smallskip

This definition corresponds to Equation~4 in~\cite{cardonha2018}. The idea,
similarly to \Cref{eq:fgs}, was to represent the difference between
$\cost(\detours\cup \{(a,b)\})$ and $\cost(\detours)$. We will show below that
it actually only represents an upper bound on this difference. Assume first that
$(a,b)$ does not intersect with a detour of $\detours$ starting on the left of
$a$. The first term corresponds to the right-hand-side of \Cref{eq:fgs} and
equals the delay added to pending files when executing the detour $(a,b)$. The
second term represents the reduction on the service time of the files in $(a,b)$
which were skipped in \detours: the time to go from $\ell(a)$ to the left of the
tape and come back, including all subsequent detours. So, in this case, it
indeed represents the intended difference.

The last sum of the definition of $\Delta$ was indexed by $f'<f$ instead of
$f'<a$ in the last line of Equation~4 of~\cite{cardonha2018}, but not on the
previous steps. Having an index $f'<f$ here would lead to an erroneously smaller
value of $\Delta$ as every detour located between $a$ and $f$ would lead to a
diminution of the value of $\Delta$, while such detours impact the service time of $f$
in the exact same way in both $\detours\cup \{(a,b)\}$ and $\detours$.

Now, assume there exists a detour $(a_1,b_1)$ in \detours such that $a_1<a$ and
$b<b_1$. Then we must have $\Delta\geq0$ as no file $f$ can be in $[a,b]$ but
not in $\detours$. Therefore, $\Delta$ does not model accurately this case,
remark which contradicts the claim in~\cite{cardonha2018} that
$\Delta(\detours,(a,b)) = \cost(\detours\cup \{(a,b)\})-\cost(\detours)$. This
fact will require to correct the algorithm \nfgs, as it relied on it to exhibit
an approximation factor of 3.

\subsection{Greedy Scheduling (\gs)}

The first algorithm proposed by \cite{cardonha2018} is named \gs for greedy
scheduling. It returns a list of all detours $(f,f)$ such that $f$ is a
requested file. It is shown to be a 3-approximation when $\uturn = 0$. Its time
complexity is $O(\nfiles)$.

\begin{algorithm}[htb]
\caption{Greedy Scheduling (\gs)}
\label{alg:gs}
\textbf{Input}: \files, \requests, \tape, \uturn\\
\textbf{Output}: A list of detours
\begin{algorithmic}[1] %
\STATE Let $\resu=\emptyset$.
\FOR{$f \in \files$}
\STATE Append $(f,f)$ to $\resu$
\ENDFOR
\STATE \textbf{return} $\resu$
\end{algorithmic}
\end{algorithm}

\subsection{Filtered Greedy Scheduling (\fgs)}

The next algorithm, \fgs, is an improvement over \gs by filtering out
detrimental detours. Such detours are determined using \Cref{eq:fgs}. As
removing a detour may lead to another detour becoming detrimental, this
subroutine is run $\nfiles$ times, for a time complexity in $O(\nfiles^2)$ as the
terms needed to evaluate \Cref{eq:fgs} can be maintained in constant time per
iteration.

\begin{algorithm}[htb]
\caption{Filtered Greedy Scheduling (\fgs)}
\label{alg:fgs}
\textbf{Input}: \files, \requests, \tape, \uturn\\
\textbf{Output}: A list of detours
\begin{algorithmic}[1] %
\STATE Let $\resu=\gs(\files,\requests,\tape,\uturn)$.
\FOR{$\_ \in \files$}
	\FOR{$(f,f) \in \resu$}
	\IF{\Cref{eq:fgs} is true}
	\STATE Remove $(f,f)$ from $\resu$
	\ENDIF
	\ENDFOR
\ENDFOR
\STATE \textbf{return} $\resu$
\end{algorithmic}
\end{algorithm}

\subsection{Non-Atomic Filtered Greedy Scheduling (\nfgs)}

The next algorithm, \nfgs~\cite{cardonha2018}, is an improvement over \fgs by replacing some
unique-file detours by more beneficial multi-files detours. Therefore, it is
claimed to also offer an approximation ratio of 3 when $\uturn=0$ as its cost
should be lower than $\gs$.

On top of the small correction on $\Delta$ described before, we also modified
the algorithm in Line~\ref{ln:argmin} and added Lines~\ref{ln:wasdetour},
\ref{ln:startif}-\ref{ln:endif}, and \ref{ln:rightdetour} in order to avoid cases in which
the cost of \fgs becomes larger than the one of \gs.

First, Line~\ref{ln:argmin}, we replaced $\arg\min_{f'> f}$ by $\arg\min_{f'\geq
f}$ as, otherwise, unique-file detours cannot be kept which increases the final
cost compared to \gs.

Then, the second issue is related to the false claim about $\Delta$. As, when
$f$ is part of a detour started on the left, the value of $\Delta$ is never
negative (and almost always positive), beneficial detours part of a longer
detour cannot be kept by the original algorithm, which increases the final cost
compared to \gs. Therefore, the added lines recognize this case and never remove
such a detour $(f,f)$ by overwriting the value of $f^*$.

This algorithm has a time complexity of $O(\nfiles^3)$, dominated by the
$O(\nfiles^2)$ evaluations of $\Delta$ which requires $O(\nfiles)$ time to be
computed.

\begin{algorithm}[htb]
\caption{Non-atomic Filtered Greedy Scheduling (\nfgs)}
\label{alg:nfgs}
\textbf{Input}: \files, \requests, \tape, \uturn\\
\textbf{Output}: A list of detours
\begin{algorithmic}[1] %
\STATE Let $\resu=\fgs(\files,\requests,\tape,\uturn)$.
\STATE Let $\rightestdetour = 0$
\FOR{$f \in \files$}
	\STATE Let $\wasdetour= \mathit{True} \textbf{ if } (f,f)\in \resu \textbf{ else } \mathit{False}$ \label{ln:wasdetour}
	\STATE Let $\tempo= \resu \setminus \{(f,f)\}$ 
    \STATE Let $f^*= \arg\min_{f'\geq f}~ (\Delta(\tempo, (f,f')))$ \label{ln:argmin}

	\IF{$\Delta(\tempo,(f,f^*))\geq 0$ and \wasdetour and $\rightestdetour>f$} \label{ln:startif}
		\STATE $f^* = f$
	\ENDIF \label{ln:endif}

	\IF{$\Delta(\tempo,(f,f^*))<0$}
		\STATE Add $(f,f^*)$ to \resu
		\STATE $\rightestdetour = \max(\rightestdetour,f^*)$ \label{ln:rightdetour}
	\ENDIF
    
\ENDFOR
\STATE \textbf{return} $\resu$
\end{algorithmic}
\end{algorithm}

\subsection{Logarithmic Non-Atomic Filtered Greedy Scheduling (\lognfgs)}

The last algorithm we present in this document is a restriction of \nfgs where
the detour lengths are bounded by $\lambda\cdot\log\nfiles$ requested files.
The original algorithm~\cite{cardonha2018} was written with a value of $\lambda=1$ but we add this
parameter for a fair comparison with \logdp.  In the experiments, we
use a parameter of $5$ as our dataset presents values of \nfiles
smaller than in the dataset used in~\cite{cardonha2018}.
Its time complexity is $O(\nfiles^2 \log \nfiles)$.

\begin{algorithm}[htb]
\caption{Logarithmic Non-atomic Filtered Greedy Scheduling (\lognfgs)}
\label{alg:lnfgs}
\textbf{Input}: \files, \requests, \tape, \uturn\\
\textbf{Parameters}: $\lambda$\\
\textbf{Output}: A list of detours
\begin{algorithmic}[1] %
\STATE Let $\resu=\fgs(\files,\requests,\tape,\uturn)$.
\STATE Let $\rightestdetour = 0$
\FOR{$f \in \files$}
	\STATE Let $\wasdetour= \mathit{True} \textbf{ if } (f,f)\in \resu \textbf{ else } \mathit{False}$ 
	\STATE Let $\tempo= \resu \setminus \{(f,f)\}$ 
    \STATE Let $f^*= \arg\min_{f'\geq f \text{ and } f'\leq f + \lambda\log \nfiles}~ (\Delta(\tempo, (f,f')))$ 

	\IF{$\Delta(\tempo,(f,f^*))\geq 0$ and \wasdetour and $\rightestdetour>f$} 
		\STATE $f^* = f$
	\ENDIF

	\IF{$\Delta(\tempo,(f,f^*))<0$}
		\STATE Add $(f,f^*)$ to \resu
		\STATE $\rightestdetour = \max(\rightestdetour,f^*)$ 
	\ENDIF
    
\ENDFOR
\STATE \textbf{return} $\resu$
\end{algorithmic}
\end{algorithm}

\clearpage

\section{Reproducibility artifact and dataset}

This section is dedicated to the reproducibility of the performance evaluation results presented in \Cref{sec:expe}. 
Section~\ref{sec:dataset} describes a dataset of reading requests on 169 tapes, associated to the 
description of all the files on these tapes. This dataset is available at {\url{https://figshare.com/s/a77d6b2687ab69416557}}. The data are extracted from real logs of a leading computing facility and is, to the best of our knowledge, the first one of its kind publicly available.
Section~\ref{sec:repro} contains all the necessary material to reproduce the simulation results presented in \Cref{sec:expe}.
This material is available in a reproducibility artifact freely accessible at {\url{https://figshare.com/s/80cee4b7497d004dbc70}}. It contains all  the instructions regarding the execution of the simulation code, the output data of the different experiments, and the scripts to generate the figures.

\subsection{A public dataset of magnetic tape file description and reading requests}
\label{sec:dataset}

In this section, we introduce a dataset containing the position and size of files on magnetic tapes, associated to user reading requests on these tapes from
a production system. The dataset is freely accessible online using the following link: {\url{https://figshare.com/s/a77d6b2687ab69416557}}.

\subsubsection{Context}
The \CC, from which our dataset is extracted, uses tape storage for long-term projects in the fields of High Energy Physics and Astroparticles Physics.
In this context, we had access to logs of the tape system from a period of high activity.
The center uses the Spectra Tfinity library, and has 48 reading engines TS1160 with 6700 Jaguar E magnetic tapes with a capacity of 20TB each.

The raw dataset covers three weeks of activity. It contains millions of lines of
reading, writing, and update requests with their associated timestamp. %
It also
details positioning operations and delays for the device heads. For obvious privacy issues, we cannot make the whole raw dataset public, but only some anonymized features.

In this work, we were interested in getting a description of magnetic tapes (position and size of files on magnetic tapes),
associated to user reading requests on these files.
The former knowledge is accessed through description files of the tapes, given by the system.
The latter is obtained from the raw logs, after several steps of filtering. 

We first removed all lines from the raw dataset that do not concern reading
operations. This gives us a list of 169 tapes, covering a total of $3,387,669$ files.
Each tape is divided into segments containing files or \emph{aggregates} of
files. The size and number of segments depend on the tape. In a segment, the files are
described by several features such as position and size. The current setup in
the computing center allows to write \emph{aggregates} of files on the
tapes, \ie a batch of related files that can be written sequentially. A
segment contains an aggregate if there is more than one file referenced in this
segment. Within an aggregate, the position of a file is given as a couple
(position,offset) here the position is actually the beginning of the aggregate,
thus the beginning of a segment. Note that an aggregate can span across several
segments. We discarded such aggregates and their associated requests to focus
on aggregates lying on a single segment. Reading files inside an aggregate is
not straightforward and generates a non-negligible overhead as the head is
required to go to the start of the aggregate before reading a file. To ease the
extraction of our sequences of requests, we considered that a requested file
inside an aggregate will be treated as a request to read the whole aggregate.
Such a behavior actually represents a strategy of buffering when aggregates are
stored on disks after a file is requested within, in order to avoid the costly
operations of accessing a file in aggregates. Thus, all the file requests in
the same aggregate are replaced by a single request for a file of the size of
this aggregate, and we associate to this file a number of requests equal to the
number of files in the aggregate.

Overall, the final processing of the logs gives us 169 tapes with a total of $119,708$ files stored on it after the filtering of aggregates, according to the tape description 
files of the system at the considered period in the logs. 
The exploitation of the system logs allowed us to extract $28,853$ unique file requests on these tapes, and a total of $615,324$ user requests over these files. 

This dataset is, to the best of our knowledge, the first publicly available dataset on magnetic tape storage.
In the next paragraphs, we describe the different files of the dataset.

\subsubsection{Characteristics of the dataset}

We provide in this section some statistics about the main characteristics of the
dataset, to illustrate the diversity of the represented instances
(tapes and associated  requests).

\begin{table}[h]
\renewcommand{\arraystretch}{1.25}
\centering
\begin{tabular}{cccc}
\toprule
&Tape size ($n_f$) &  \# Files Requested ($\nfiles$) & \# Total User Requests ($n$)\\
\midrule
Maximum & 4,142  & 852  & 15,477  \\
Minimum & 111  & 31 & 1,182 \\
Median & 490 & 148  & 2,669 \\
Mean & 709 & 170 & 3,640 \\
\bottomrule
\end{tabular}
\caption{Overview of the instances characteristics related to the number of files.}
\label{tab:sum}
\end{table}

\begin{figure}[htb]
\centering
\resizebox{0.6\linewidth}{!}{
%
}
\caption{Illustration of the tape dataset with the number of files in each tape in function of the number of unique requested files in it.}
\label{fig:tape_file}
\end{figure}

\begin{figure}[htb]
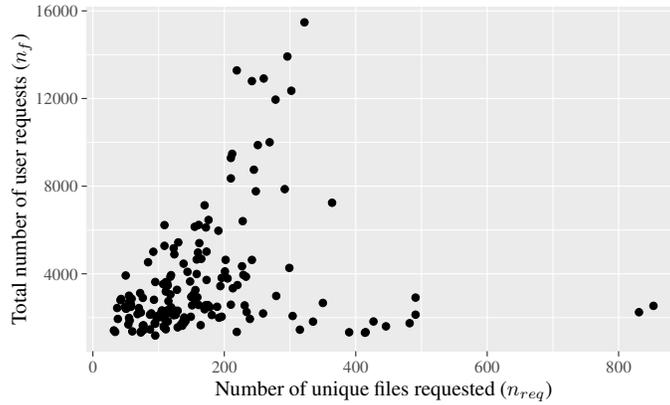

\centering
\resizebox{0.6\linewidth}{!}{
%
}
\caption{Illustration of the tape dataset with the number of unique requested files in each tape in function of the total number of user requests in it.}
\label{fig:tape_req}
\end{figure}

\paragraph{Statistics on the number of files and requests.}
Table~\ref{tab:sum} gives a brief summary of the dataset in terms of tape size and number of requests. 
There is a large variety of tape sizes, from hundreds to thousands of files. The same observation stands for the number of files requested and the total number of requests on those files.
Figure~\ref{fig:tape_file} represents the distribution of unique files requested in function of the size of the tapes. Most tapes consist of less than a thousand of files and have at most 300 unique files requested, and there is no strong visible correlation between these parameters, which ensures the diversity of the dataset.
We display in Figure~\ref{fig:tape_req}, for each tape, the distribution of the total number of user requests with the number of unique files requested. 
We also observe that the total number of user requests is varied even among tapes having a very similar number of unique files requested.

\paragraph{Statistics on the sizes of the files.}

We now focus on the distribution of file sizes among the tapes. \Cref{tab:sizes}
first shows the statistical summary of the average file size in a tape, ranging
from 5 to 167GB with an average of 50GB. This information is slightly redundant
as usually proportional to $1/n_f$, most tapes being full and of the same
capacity. The important information provided here concerns the coefficient of
variation of the file sizes in each tape ({\it i.e.,} the standard deviation
over the average file size in a tape, expressed as a percentage). We can see that
many tapes present varied file sizes, as the median coefficient of variation
equals 56\% and the average is 94\%. This corresponds to more difficult instances
of the targeted problem, as greedy solutions are sufficient to solve the problem
with a variance of 0 and no request multiplicity. \Cref{fig:tape_sd} shows the
relation between the mean file size and the coefficient of variation: a larger
mean file size (hence a smaller $n_f$) is related to lower coefficients of
variation, but again there is no direct dependency and a few clusters can be
identified in this plot.

\bigskip
We therefore believe this dataset is heterogeneous and suitable for performance
evaluation of a magnetic tape storage system.

\begin{table}[htbp]

\renewcommand{\arraystretch}{1.25}
\centering
\begin{tabular}{ccc}
\toprule
  & Average file size (GB) & File size coefficient of variation\\
\midrule
Maximum  & 167 & 379\%\\

Minimum  & 4.9& 6\% \\
Median  & 40 & 56\%\\
Mean &  50 & 94\%\\
\bottomrule
\end{tabular}
\caption{Overview on the file sizes present in each tape.}
\label{tab:sizes}
\end{table}

\begin{figure}[htbp!]
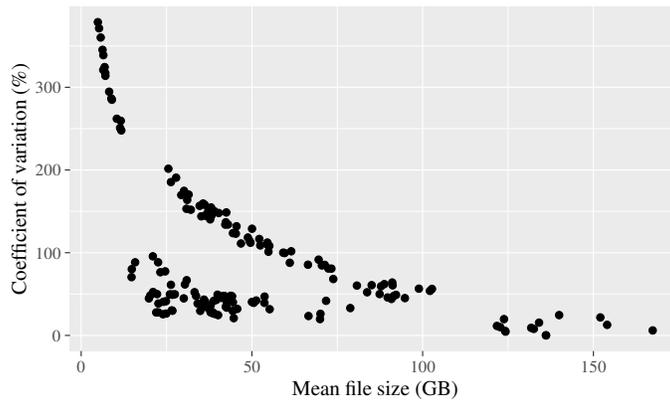

\centering
\resizebox{0.6\linewidth}{!}{
%
}
\caption{Illustration of the tape dataset with the file sizes coefficient of variation in each tape in function of the average file size of the tape.}
\label{fig:tape_sd}
\end{figure}

\subsubsection{Dataset content}

We now describe the content of the public folder.

\paragraph{`\textbf{list\_of\_tape.txt}'}
This file lists the name of the 169 tapes in the dataset. For each tape, there is a file listing all the user requests on this tape in the folder \textbf{requests}, and 
a file describing the content of the tape in the folder \textbf{tapes}.
The tapes are named under the format \textit{TAPEXXX.txt} where XXX varies from 001 to 169.

\paragraph{\textbf{requests} folder}
For each tape, this folder contains a request file with two columns
\verb|index| and \verb|nb_requests|. The former refers to the index of the requested file 
on the tape (see \textbf{tapes} folder) associated to the number of requests for this file. 
The maximum number of distinct files requested for one tape is equal
to 852, and the minimum number is 31. The median value is 148 unique
files (for a tape with 531 files), and the mean is 170.
Regarding the total number of user requests on one tape, the maximum
is 15,477 and the minimum is 1,182, for a median value of 2,669 files
and a mean of 3,640. 

\paragraph{\textbf{tapes} folder}
This folder contains a description file of each tape in the dataset. From the left (position 0) to the right of each tape, the file describes the different segments
of the tape.\\ It contains four columns \verb|id,cumulative_position,segment_size,index|. The \verb|id| column corresponds to the id number of the segment on the tape given by the system. The next two columns respectively refer to the cumulative position of the segment from the left of the tape, and its size. Finally, the \verb|index| column is used as the id of the file on the tape starting from 1 for the leftmost file. This fourth column is used to match the \verb|index| column of the \textbf{requests} files.
The largest tape contains 4,141 files, and the smallest one 111.
The median size is 489 files and the mean size is 708 files.

\subsubsection{Perspectives}

This dataset allowed us to evaluate several algorithms on realistic data extracted from the logs of a production computing center.
We expect this dataset to be a first step in the achievement of large-scale datasets of such types.
Logs from a larger time period can be envisioned as an extension to this dataset.

In this work, we only considered reading requests from users in the framework of the Linear Tape Scheduling Problem. However, the raw logs contains much more information
that one could expect to use. Knowledge about time processing of reading operations and positioning operations performed by the multiple device heads could be leveraged to better
model seeking speed and reading speed. A rapid overview of the logs
tends to show that the positioning time seems to impact the performance much more than the reading time. Hence, modeling 
the seeking speed of the device seems to be important to provide realistic cost models of the process.
Temporal aspects of the raw dataset could also be exploited for a usage in online problems, for instance.

\subsection{Reproducibility artifact}
\label{sec:repro}

This section provides all the details to reproduce the performance evaluation presented in \Cref{sec:expe}.
The complete artifact can be downloaded online: {\url{https://figshare.com/s/80cee4b7497d004dbc70}}.
\subsubsection{`\textbf{input}' folder}

This folder contains the data described in Section~\ref{sec:dataset}. The reader is invited to refer to this section for comprehensive details about the dataset used 
for the performance evaluation, and how it has been generated.
The folder \textbf{requests} contains the index of the files requested on a tape, associated to the number of requests of this file.
The folder \textbf{tape} describes the position and size of the files
on a tape. Both folders are used as input of the differnt algorithms
presented in the paper (see \textbf{code} folder).

\subsubsection{`\textbf{code}' folder}
This folder contains a Python implementation of our algorithms and of those adapted from~\cite{cardonha2018} used for baseline comparison.
We carefully implemented the different strategies in  the  \textit{algorithms.py} file.
The \textit{main.py} file is dedicated to the execution of all algorithms on all the instances of the \textbf{input} folder.
It directly parses the different files in the \textbf{input} folder to instantiate 4 different parameters of the algorithms:
\begin{itemize}
\item \textbf{files\_requested}: the list of requested files on the tape, comes from the \verb|index| column in the \textit{input/requests/TAPEXXX.txt} files.
\item \textbf{request\_numbers}: the number of requests of each file in the above list. Extracted from the \verb|nb_requests| column in the \textit{input/requests/TAPEXXX.txt} files.
\item \textbf{tape}: the list of all file sizes on the tape. Extracted from the \verb|segment_size| column of the \textit{input/tapes/TAPEXXX.txt} files/
\item \textbf{right}: a list of the right ordinate of each file in \textbf{tape}. Obtained by computing the cumulative sum of the \textbf{tape} parameter
\end{itemize}

We also provide in the \textit{draw.py} file a visualization tool of the device head trajectory depending on the list of detours produced by the algorithms.
This tool is automatically called in \textit{main.py} for each input
and algorithm pair.

To start the performance evaluation, one should just go into the \textbf{code} folder, and start the program using the makefile:

\begin{lstlisting}[language=bash,frame=single]
cd code ; make
\end{lstlisting}

It requires to have  \verb+python3+ installed on the machine. It can easily be installed on any Ubuntu/Debian machine using the following command
\begin{lstlisting}[language=bash,frame=single]
sudo apt-get install python3
\end{lstlisting}

The performance evaluation in \Cref{sec:expe} uses Python3 version 3.9.2.
The code has been executed on a compute node with two Intel Xeon Gold 6130 CPUs with 16 cores each.
The execution of the algorithms has been performed sequentially on a single core of a dedicated node to avoid external disturbances. 

\subsubsection{`Run' folder}

This folder contains the performance results of the different strategies evaluated in \Cref{sec:expe} of the paper.
For each algorithm, we recorded the cost induced by the list of detours in output and the simulation time to get the solution.
We tested three different values of the U-turn penalty, that is a parameter:
\begin{itemize}
\item 0: no penalty
\item 14,254,750,000: it represents half of the average size of a tape segment according to our 169 input tapes.
\item 28,509,500,000: it represents the average size of a tape segment according to our 169 input tapes.
\end{itemize}

The \textit{results.csv} file summarizes the cost of the list of detours induced by each algorithm, associated to the time-to-solution to get this 
list for each of the three penalties above presented. We also record the lower bound for each algorithm on each input.

\subsubsection{`Figure' folder}
This folder contains a R script that processes the \textit{run/results.csv} to reproduce the figures presented in \Cref{sec:expe} of the paper.

\end{document}